\documentclass[11pt]{article}
\usepackage[letterpaper, portrait, margin=1in]{geometry}
\usepackage[utf8]{inputenc}
\usepackage{nicefrac}
\usepackage{natbib}

\usepackage{float}

\usepackage{amsmath, mathtools, amsthm, amssymb, thm-restate}
\usepackage{algorithm, algpseudocode, yhmath}

\usepackage[dvipsnames]{xcolor}

\usepackage{hyperref}

\usepackage[textwidth=4em,textsize=tiny]{todonotes}

\theoremstyle{plain}
\newtheorem{theorem}{Theorem}
\newtheorem{lemma}[theorem]{Lemma}
\newtheorem{fact}[theorem]{Fact}
\newtheorem{corollary}[theorem]{Corollary}
\newtheorem{proposition}[theorem]{Proposition}

\theoremstyle{definition}
\newtheorem{definition}[theorem]{Definition}

\theoremstyle{remark}

\usepackage[most]{tcolorbox}
\newtcolorbox{idea}[1][]
{
colbacktitle=cyan,
colback=cyan!10,
arc=1pt,
boxrule=1pt,
title=#1 
}

\newtcolorbox{update}[1][]
{
colbacktitle=gray,
colback=gray!10,
arc=1pt,
boxrule=1pt,
title=#1 
}

\newtcolorbox{question}[1][]
{
coltitle=black,
colbacktitle=yellow,
colback=yellow!10,
arc=1pt,
boxrule=1pt,
title=#1 
}

\newtcolorbox{note}[1][]
{
coltitle=black,
colbacktitle=green,
colback=green!10,
arc=1pt,
boxrule=1pt,
title=#1 
}

\newcommand{\be}{\boldsymbol{e}}

\newcommand{\bA}{\boldsymbol{A}}

\newcommand{\bT}{\boldsymbol{T}}

\newcommand{\sA}{\mathcal{A}}

\newcommand{\sD}{\mathcal{D}}

\newcommand{\sL}{\mathcal{L}}

\newcommand{\sT}{\mathcal{T}}

\newcommand{\iscq}{\texttt{ISC}}
\newcommand{\cisq}{\texttt{cIS}}
\newcommand{\ciscq}{\texttt{cISC}}

\newcommand{\graph}{G}
\newcommand{\sgraph}{H}
\newcommand{\vertexset}{V}
\newcommand{\edgeset}{E}
\newcommand{\vertexcount}{n}
\newcommand{\edgecount}{m}
\newcommand{\vertex}{v}
\newcommand{\altvertex}{u}
\newcommand{\edge}{e}
\newcommand{\neighbour}{{\sf N}}

\newcommand{\samplededges}{S}

\newcommand{\approxedgecount}{\widehat{\edgecount}}

\newcommand{\collisioncount}{r}

\newcommand{\iscollision}{\texttt{IsCollision}}
\newcommand{\collision}{k}
\newcommand{\preresamplesize}{t}
\newcommand{\birthdayestimate}{\texttt{BirthdayEstimate}}
\newcommand{\ensemblecolor}{\texttt{EnsembleColor}}
\newcommand{\ensemblecolorpartite}{\texttt{EnsembleColorPartite}}
\newcommand{\approxsubgraph}{\texttt{ApproxSubhypergraph}}
\newcommand{\mulestimation}{\texttt{MulEstimation}}
\newcommand{\fractionofedges}{p}
\newcommand{\approxfractionofedges}{\widehat{\fractionofedges}}
\newcommand{\edgeestimationISC}{\texttt{EdgeEstimationPartite}}
\newcommand{\edgeestimation}{\texttt{EdgeEstimation}}
\newcommand{\approxsubseq}{\texttt{ApproxSubsequence}}
\newcommand{\threshold}{\tau}
\newcommand{\iscoll}{{coll}}
\newcommand{\iterations}{\kappa}
\newcommand{\thresholdconst}{\zeta}
\newcommand{\subsequence}{\texttt{subseeq}}
\newcommand{\event}{\mathcal{E}}
\newcommand{\resamplesize}{s}

\newcommand{\birthdayevent}{\mathcal{E}_\mathrm{bir}}
\newcommand{\subseqevent}{\mathcal{E}_\mathrm{sub}}

\newcommand{\qsetno}{s}
\newcommand{\hyperlen}{d}
\newcommand{\instances}{t}
\newcommand{\onto}[2]{f(#1, #2)}

\newcommand{\dyes}{\sD_{\mathrm{yes}}}
\newcommand{\dno}{\sD_{\mathrm{no}}}

\newcommand{\ocvset}{\mathbf{K}}
\newcommand{\oclbl}{\boldsymbol{\ell}}
\newcommand{\oceset}{\be}

\newcommand{\lbthres}{\threshold}

\newcommand{\alg}{\sA}
\newcommand{\augalg}{\overline{\alg}}
\newcommand{\numq}{q}

\newcommand{\oracle}{\fbrac{\ocvset,\oclbl,\oceset}}

\newcommand{\field}[1]{\mathbb{#1}}
\newcommand{\R}{\field{R}}
\newcommand{\Nat}{\field{N}}

\newcommand{\func}[3]{{#1} : {#2} \rightarrow {#3}}

\newcommand{\sequence}[2]{{#1}_1,{#1}_2,...,{#1}_{#2}}
\newcommand{\fbrac}[1]{\left({#1}\right)}
\newcommand{\sbrac}[1]{\left\{{#1}\right\}}
\newcommand{\tbrac}[1]{\left[{#1}\right]}
\newcommand{\ang}[1]{\langle{#1}\rangle}
\newcommand{\abs}[1]{\left|{#1}\right|}
\newcommand{\size}[1]{\left|{#1}\right|}
\newcommand{\constant}{c}

\DeclareMathOperator*{\E}{\field{E}}
\DeclareMathOperator{\Var}{\mathrm{Var}}

\newcommand{\indicator}{\mathbf{1}}

\newcommand{\approxerror}{\varepsilon}
\newcommand{\confidence}{\delta}

\newcommand{\smallo}[1]{o\fbrac{{#1}}}

\newcommand{\bigomega}[1]{\Omega\fbrac{{#1}}}

\newcommand{\poly}[1]{\mathrm{poly}\fbrac{#1}}

\usepackage{cleveref}
\usepackage{booktabs}   
\usepackage{array}      
\usepackage{tabularx}   
\usepackage{makecell}   
\usepackage{soul}
\title{Power of Conditional Sampling on Hypergraphs}
\author{
Debarshi Chanda
\footnote{Indian Statistical Institute, Kolkata, India}
\and 
Buddha Dev Das 
\footnotemark[1]
\and 
Arijit Ghosh
\footnotemark[1]
\and 
Gopinath Mishra
\footnote{
Institute of Mathematical Sciences, Chennai, India
}
}

\date{}

\begin{document}

\maketitle
\thispagestyle{empty}

\begin{abstract}
Let $G=(V,E)$ be a $d$-uniform hypergraph with $|V|=n$ and $|E|=m$. 
Estimating $m$ from uniformly random hyperedge samples requires $\Theta(\sqrt{m})$ samples by a 
standard birthday-paradox argument. We investigate the power of conditional sampling in this setting. 
We consider the following oracle. Given a subset $X\subseteq V$, the oracle returns a hyperedge chosen 
uniformly at random from all hyperedges induced by $X$, provided that at least one such hyperedge 
exists. More generally, for an integer $1\le s\le d$, we consider a hierarchy of conditional sampling
oracles: given $s$ pairwise disjoint subsets $X_1,\ldots,X_s\subseteq U$, the oracle returns a hyperedge 
chosen uniformly at random from all hyperedges that are contained in $X_1\cup\cdots\cup X_s$ and 
intersect each set $X_i$, $i\in[s]$, provided that at least one such hyperedge exists.

For every $1\le s\le d$, we give a randomized algorithm that outputs
an estimate $\widehat{m}$ satisfying
$\widehat{m}\in[(1-\varepsilon)m,(1+\varepsilon)m]$ with probability
at least $2/3$, using only
$O_{s,d}(\operatorname{poly}(1/\varepsilon,\log n))$ conditional
samples. Note that the exponent of the $\log n$ factor in the query
complexity is an absolute constant, independent of both $s$ and $d$.
Consequently, conditional sampling yields an exponential separation
in query complexity over algorithms that rely solely on uniformly
random hyperedge samples.

We also study the corresponding decision version, in which the oracle
only reports whether a qualifying hyperedge exists. For every
$1\le s<d$, we prove that any algorithm using the $s$-colorful
decision oracle requires polynomially many queries, namely
$\widetilde{\Omega}_{s,d}\bigl(n^{\frac{d-s}{d-s+1}}\bigr)$. Thus, for every
$s<d$, access to uniformly random witness hyperedges (random
certificates) is exponentially more powerful than mere existence
queries. When $s=d$, one could achieve the query complexity with the decision oracle
also becomes $\operatorname{poly}(1/\varepsilon,\log n)$, but the
exponent of the $\log n$ factor is known to be $\Theta(d)$. The colorful hierarchy
thus exhibits a phase transition at $s=d$: for $s<d$, random
certificates are essential for polylogarithmic-query estimation,
whereas at $s=d$, decision queries alone suffice, yet random
certificates still provably reduce the exponent of the $\log n$
factor from $\Theta(d)$ to an absolute constant.

Let $G=(V,E)$ be a $d$-uniform hypergraph with $|V|=n$ and $|E|=m$. Estimating $m$ from uniformly random hyperedge samples requires $\Theta(\sqrt{m})$ samples by a birthday-paradox argument. We investigate the power of conditional sampling through a hierarchy of oracles: for an integer $1\le s\le d$, given pairwise disjoint subsets $X_1,\ldots,X_s\subseteq V$, the $s$-colorful oracle returns a hyperedge chosen uniformly at random from all hyperedges contained in $X_1\cup\cdots\cup X_s$ that intersect every $X_i$, provided one exists. For $s=1$, this samples a uniformly random hyperedge induced by $X_1$.

For every $1\le s\le d$, we give a randomized algorithm that outputs $\widehat{m}\in[(1-\varepsilon)m,(1+\varepsilon)m]$ with probability at least $2/3$ using $O_{s,d}(\operatorname{poly}(1/\varepsilon,\log n))$ conditional samples, where the exponent of $\log n$ is an absolute constant independent of both $s$ and $d$. Conditional sampling thus yields an exponential separation over uniform hyperedge sampling.

We also study the decision version, in which the oracle only reports whether a qualifying hyperedge exists. For every $1\le s<d$, we prove that any algorithm using the $s$-colorful decision oracle requires $\widetilde{\Omega}_{s,d}(n^{\frac{d-s}{d-s+1}})$ queries. Hence, for $s<d$, access to uniformly random witness hyperedges (random certificates) is exponentially more powerful than existence queries. For $s=d$, decision queries alone achieve $\operatorname{poly}(1/\varepsilon,\log n)$ query complexity, but the exponent of $\log n$ is known to be $\Theta(d)$. The colorful hierarchy thus exhibits a phase transition at $s=d$: for $s<d$, random certificates are essential for polylogarithmic-query estimation, whereas at $s=d$ decision queries suffice, yet random certificates still provably reduce the exponent of $\log n$ from $\Theta(d)$ to an absolute constant.
\end{abstract}

\paragraph{AI use statement.}
AI tools were used solely to improve the paper's readability. All mathematical ideas and proofs presented in this work are the authors' own. The authors take full responsibility for the content of the submitted manuscript.



\newpage
\thispagestyle{empty}
\tableofcontents
\newpage
\setcounter{page}{1}

\section{Introduction}

 Edge estimation, and its generalization to hyperedge estimation via query oracles, is one of the central problems in query complexity and property testing, and has been studied extensively over the past two decades
\cite{feige04, avg_deg_danaron08, BerettaT24, beretta2025, BishnuCDGM25, DBLP:journals/toct/DellL21, DellLM22, DellLM24, harpeled18, BhattacharyaBGM24, CDGM26, addanki_et_al:LIPIcs.ESA.2022.2, Xichen19, AdarHotamLevi/Arxiv/2601.21457,Censor-HillelEW25}. Besides being a natural computational primitive, hyperedge estimation has led to the development of several powerful (algorithmic) techniques. Consequently, it has become a standard benchmark for understanding the power of different query oracles.

Let $G=(V,E)$ be a $d$-uniform hypergraph with $|V|=n$ and $|E|=m$. The \emph{hyperedge estimation problem} asks to estimate the number of hyperedges $m$ using query access to $G$. More precisely, given an approximation parameter $\varepsilon\in(0,1)$, the goal is to design a randomized algorithm that outputs an estimate $\widehat{m}$ satisfying
$
(1-\varepsilon)m \le \widehat{m} \le (1+\varepsilon)m
$
with constant success probability while minimizing the number of queries made to the underlying oracle.

The query complexity of hyperedge estimation depends crucially on the type of oracle available to the algorithm. If the algorithm has access only to uniformly random hyperedge samples, then estimating $m$ requires $\widetilde{\Theta}(\sqrt{m})$ samples by the classical birthday-paradox argument~\cite{Mitzenmacher_Upfal_2005}. Using \footnote{ $\widetilde{\Theta}(\cdot)$, we suppress $\poly{\log\vertexcount,\approxerror}$ factors, and Similarly, we define $\widetilde{O}(\cdot)$ and $\widetilde{\Omega}(\cdot)$.} Consequently, obtaining substantially more efficient algorithms requires access to richer query primitives.

Among the many query models studied in the literature, \emph{conditional
sampling} has emerged as one of the most powerful and versatile
primitives~\cite{ChakrabortyFischerGoldhirshMatsliah/ITCS/2013/CondSamp,DBLP:journals/siamcomp/CanonneRS15/Conditional,CanonneRS14,BhattacharyyaC18,CanonneCKLW21,Bhattacharyya0P24}. Instead of sampling a uniformly random hyperedge from the entire hypergraph, a conditional sampling oracle allows the algorithm to specify a collection of structural constraints and returns a uniformly random hyperedge satisfying those constraints. 


\subsection{Models and our results}

A variety of \emph{global query oracles} have been studied for graph and hypergraph estimation. One of the natural global query oracle for graphs is the \emph{independent set} (IS) oracle \cite{harpeled18, Xichen19, AdarHotamLevi/Arxiv/2601.21457}, which, given a subset of vertices, reports whether the induced subgraph contains an edge. For graphs, a more powerful query oracle is the \emph{bipartite independent set} (BIS) oracle \cite{DBLP:journals/toct/DellL21, harpeled18}, which, given two disjoint subsets of vertices, reports whether there exists an edge with one endpoint in each subset. Very recently, \cite{CDGM26} introduced the \emph{independent set oracle with random certificates}, which additionally returns a uniformly random witness edge whenever one exists. For general $d$-uniform hypergraphs, \cite{BhattacharyaBGM19, DellLM22, BhattacharyaBGM24} considered the \emph{$d$-colorful decision oracle}, which, given $d$ pairwise disjoint subsets $V_1,\ldots,V_d$, reports whether there exists a hyperedge intersecting every subset.

These oracle models have proved remarkably powerful for edge and hyperedge estimation. However, despite these developments, there has been no systematic study of the hierarchy of colorful query oracles. Existing work focuses on a few isolated points in this hierarchy—namely the IS oracle, the BIS oracle, the IS oracle with random certificates, and the $d$-colorful decision oracle—leaving the intermediate colorful query models essentially unexplored.






Motivated by this gap, as well as developments in the property testing of probability distributions through the lens of conditional sampling \cite{ChakrabortyFischerGoldhirshMatsliah/ITCS/2013/CondSamp,DBLP:journals/siamcomp/CanonneRS15/Conditional,BhattacharyyaC18}, we introduce a unified hierarchy of \emph{$s$-colorful} query oracles, parameterized by an integer $1\le s\le d$. Given pairwise disjoint subsets $V_1,\ldots,V_s\subseteq V$, the
\emph{$s$-colorful random certificate oracle}, denoted by $\ciscq(V_1,\ldots,V_s)$, returns a uniformly random hyperedge that is
contained in $V_1\cup\cdots\cup V_s$ and intersects every subset $V_i$,
$i\in[s]$, provided that such a hyperedge exists. We also consider the corresponding \emph{$s$-colorful decision oracle}, denoted by $\cisq(V_1,\ldots,V_s)$, which simply reports whether such a hyperedge exists.  The hierarchy unifies all previously studied oracle models: for graphs ($d=2$), $\cisq{}$ for $s=1$ and $s=2$ corresponds to the IS and BIS oracles, respectively, while $\ciscq(V_1)$ coincides with the independent set oracle with random certificates. For general $d$, the case $s=d$ coincides with the previously studied $d$-colorful decision oracle.

Our work is motivated by the following three fundamental questions.

\begin{enumerate}
\item[\textbf{Q1.}] How much more powerful is the $s$-colorful random certificate oracle $\ciscq$ than uniformly random hyperedge sampling for hyperedge estimation?

\item[\textbf{Q2.}] For a fixed $1\le s\le d$, how much additional algorithmic power does the $s$-colorful random certificate oracle $\ciscq{}$ provide over the corresponding decision oracle $\cisq{}$?

\item[\textbf{Q3.}] How does the relative power of $\ciscq{}$ and $\cisq{}$ evolve as $s$ increases from $1$ to $d$? In particular, does a qualitative change occur at some $s$?
\end{enumerate}

We answer all three questions. Our first result shows that a polylogarithmic number of $s$-colorful random certificate queries suffices to estimate the number of hyperedges.
\begin{restatable}[]{theorem}{mainUB}\label{thm:main-intro}
For every integer $1\le s\le d$ and $\approxerror \in (0, 1/3)$, there exists a randomized algorithm that, given access to the $s$-colorful random certificate oracle $\ciscq$, outputs a $(1\pm\varepsilon)$-approximation of the number of hyperedges with probability at least $2/3$ using $O_{s,d} \fbrac{\varepsilon^{-4}\log^2 n\log\log n\log^2(1/\varepsilon)}$ $s$-colorful random certificate queries in expectation. \footnote{By $O_{s,d}(\cdot)$, we denote $\constant_{s,d}O(\cdot)$, where $\constant_{s,d}$ is a constant depending only on $s$ and $d$. We also define $\Omega_{s,d}$ analogously.}
\end{restatable}



 \Cref{thm:main-intro} answers \textbf{Q1} in the affirmative and demonstrates the power of conditional sampling for hyperedge estimation. While estimating the number of hyperedges from uniformly random hyperedge samples requires $\widetilde{\Theta}(\sqrt{m})$ queries by the classical birthday-paradox argument, our algorithm uses only a polylogarithmic number of $s$-colorful random certificate queries. This yields an exponential separation between the query complexity of hyperedge estimation using uniformly random hyperedge sampling and using $s$-colorful random certificate queries. We note that our algorithm treats all values of $s$ and $d$ in a unified manner.


Since every $s$-colorful random certificate query reveals at least as much information as the corresponding decision query, it is immediate that $\ciscq$ is at least as powerful as $\cisq$. A natural question is whether this additional information leads to a genuine algorithmic advantage. The following theorem answers \textbf{Q2} by showing that, for every $1\le s<d$, any algorithm using only the decision oracle $\cisq$ must make a polynomial number of queries to estimate the number of hyperedges.

\begin{theorem}\label{thm:lower-intro}
Let $1\le s<d$. Any randomized algorithm that, given access only to the $s$-colorful decision oracle $\cisq$, outputs a $(1\pm\varepsilon)$-approximation of the number of hyperedges with probability at least $2/3$ must make
$\widetilde{\Omega}_{s,d}\!\left(n^{\frac{d-s}{d-s+1}}\right)$
$s$-colorful decision queries in the worst case.
\end{theorem}

\Cref{thm:lower-intro}, together with \Cref{thm:main-intro}, establishes an exponential separation between the $s$-colorful random certificate oracle $\ciscq$ and the corresponding decision oracle $\cisq$ for every $1\le s<d$.

A natural question is whether this exponential separation persists for all values of $s$. The answer is no. When $s=d$, \cite{BhattacharyaBGM24,DellLM22} showed that the number of hyperedges can  be estimated using only $O(\varepsilon^{-2} \log^{O(d)}n)$ $d$-colorful decision queries. Thus, unlike the case $s<d$, random certificates are no longer necessary to obtain a polylogarithmic-query algorithm, albeit with a dependence on $d$ in the exponent of the $\log n$ factor. Combined with our lower bound for $s<d$,  the query complexity of hyperedge estimation in the decision model undergoes a transition from polynomial to polylogarithmic precisely at $s=d$, thereby providing an answer to \textbf{Q3}.

\subsection{Comparison with prior work}
\label{ssec:comparison_with_prior_work}

Our work also fits within the framework of hyperedge estimation under the \emph{global query access} model, extending a growing body of work on graph and hypergraph estimation using global query oracles
\cite{DBLP:journals/toct/DellL21, DellLM22, DellLM24, harpeled18, BhattacharyaBGM24, CDGM26, addanki_et_al:LIPIcs.ESA.2022.2, Xichen19, AdarHotamLevi/Arxiv/2601.21457}.

For graphs ($d=2$), \cite{harpeled18,Xichen19} studied the \emph{independent set} (IS) oracle, which coincides with our $1$-colorful decision oracle $\cisq$. They showed that estimating the number of edges requires $\widetilde{\Theta}(\sqrt{n})$ IS queries. Subsequently, \cite{harpeled18,addanki_et_al:LIPIcs.ESA.2022.2} considered the more expressive \emph{bipartite independent set} (BIS) oracle, corresponding to our $2$-colorful decision oracle. The latter work gave a $\widetilde{{O}}(\varepsilon^{-5}\log^5 n)$-query algorithm for edge estimation.

For general $d$-uniform hypergraphs, \cite{BishnuGKMS23,DellLM22}
introduced the \emph{$d$-colorful decision oracle} $\cisq$. Over a series of
works, increasingly efficient
algorithms were developed for hyperedge estimation using $d$-colorful
decision queries, culminating in the
$O_d(\varepsilon^{-2}\log^{3d+5} n)$-query algorithm
of~\cite{DellLM22,BhattacharyaBGM24}. Subsequently, Dell, Lapinskas, and
Meeks~\cite{DellLM24} proved a lower bound of
$\Omega_d\bigl(\log^{d(1-o_d(1))} n\bigr)$ $d$-colorful decision
queries, matching the upper bound up to a constant factor in the
exponent of $\log n$. Consequently, at $s=d$, the exponent of the
$\log n$ factor in the decision model is $\Theta(d)$, whereas the
exponent in \Cref{thm:main-intro} is an absolute constant.

Very recently, \cite{CDGM26} introduced the random certificate analogue of the IS oracle, which is precisely our $1$-colorful random certificate oracle $\ciscq$, and showed that edge estimation can be performed using only $\widetilde{O}(\varepsilon^{-2}\log^2 n)$ queries. Together with the $\widetilde{\Omega}(\sqrt{n})$ lower bound for the IS oracle due to \cite{Xichen19}, this establishes an exponential separation between the random certificate and decision variants in the graph setting.

\begin{table}[ht]
\centering
\begin{tabular}{llll}
\toprule
Oracle & Setting & Query complexity & Reference\\
\midrule
IS decision & Graphs & $\widetilde{\Theta}(\sqrt n)$ & \cite{Xichen19}\\
BIS decision & Graphs & $\widetilde{O}(\log^5 n)$ & \cite{addanki_et_al:LIPIcs.ESA.2022.2}\\
IS random certificate & Graphs & $\widetilde{O}(\log^2 n)$ & \cite{CDGM26}\\
\midrule
$s$-colorful decision & $ s=1$ & $\widetilde{\Theta}_{d}\!\left(n^{{d}/{4}}\right)$ & \cite{DellLM24}\\
$s$-colorful decision & $1\le s<d$ & $\widetilde{\Omega}_{s,d}\!\left(n^{\frac{d-s}{d-s+1}}\right)$ & \Cref{thm:lower-intro}\\
$d$-colorful decision & $s=d$ & $O_d\!\left(\log^{3d+5} n\right)$ & \cite{DellLM22,BhattacharyaBGM24}\\
$d$-colorful decision & $s=d$ & $\Omega_d\!\left(\log^{d(1-o_d(1))} n\right)$ & \cite{DellLM24}\\
$s$-colorful random certificate & $1\le s\le d$ & $\widetilde{O}_{s,d}\!\left(\log^{2} n\right)$ & \Cref{thm:main-intro}\\
\bottomrule
\end{tabular}
\caption{Comparison of our results with prior work for hyperedge
estimation. The dependence on $\varepsilon$ is suppressed throughout;
$\widetilde{O}(\cdot)$ hides $\operatorname{poly}(\log\log n)$
factors.}
\label{table:compare}
\end{table}

\Cref{table:compare} compares our results with prior work. Our work unifies these developments. We provide the first systematic study of the hierarchy of colorful query oracles. On the algorithmic side, we establish polylogarithmic-query algorithms for every $s$-colorful random certificate oracle $\ciscq$, thereby generalizing the recent result of \cite{CDGM26}. On the lower-bound side, we prove polynomial lower bounds for the corresponding $s$-colorful decision oracles $\cisq$ whenever $s<d$, generalizing the $\widetilde{\Omega}(\sqrt{n})$ lower bound for the IS oracle. Finally, we identify a sharp phase transition at $s=d$, where the query complexity of hyperedge estimation in the decision model drops from polynomial to polylogarithmic. 


We remark that, for $s=1$, \cite{DellLM24} established a stronger lower bound of $\widetilde{\Omega}(n^{d/4})$. This matches our bound in \Cref{thm:lower-intro} for graphs and $3$-uniform hypergraphs, while being stronger for larger values of $d$. Nevertheless, our lower bound applies to the entire hierarchy of colorful decision oracles, ruling out polylogarithmic-query algorithms even when $s=d-1$.



A natural question arising from our work is whether the \emph{uniform
randomness} of the returned certificate is essential for obtaining
polylogarithmic-query algorithms when $s<d$. Our algorithm relies crucially
on the uniformity of the returned hyperedge to obtain unbiased estimates,
and this reliance is provably necessary: an oracle that returns an
\emph{arbitrary} certificate can be simulated using only $O_{s,d}(\log n)$ $s$-colorful decision ($\cisq{}$) queries by binary search, and
hence, by the lower bound in \Cref{thm:lower-intro}, cannot yield a
polylogarithmic-query algorithm for any $s<d$. Thus, the exponential speedup
of \Cref{thm:main-intro} comes from the distribution of the returned
witness. Interestingly, uniformity of the witness does not always confer
such power: for approximate model counting, a SAT oracle returning a
uniformly random satisfying assignment is provably no more powerful than an
NP oracle~\cite[Theorem~1.2]{ChakrabortyCKM/ICALP/2023/SATvsNPOracle}. We
return to this comparison through the lens of fine-grained complexity at
the end of \Cref{ssec:connections_to_related_areas}.

\subsection{Connections to related areas}
\label{ssec:connections_to_related_areas}

Our work is closely connected to several well-studied themes in query complexity and sublinear algorithms. The relationship between global query oracles and \emph{group testing} has long been recognized, with the IS, BIS, and colorful decision oracles naturally viewed as structured group tests over the edge (or hyperedge) set \cite{Dorfman/AnnalsofMathStat/1943/GroupTestingFirst,AldridgeJS/Book/2019/GroupTestingIT,harpeled18,DellLM22}. In this viewpoint, hyperedges play the role of defective items, while the underlying vertex sets impose structural constraints on the families of hyperedges that can be queried. Our hierarchy of colorful query oracles provides a unified framework for studying these structured group-testing primitives and their random-certificate counterparts.

\paragraph{Decision vs Counting:} Our work is also related to the classical theme of \emph{decision versus counting}, initiated by the seminal work of \cite{Stockmeyer/stoc/83/ComplxtyApprxCnt,Stockmeyer/siamcomp/85/ApprxAlgoforSharpP}, which showed that approximate counting can often be reduced to decision queries. In our setting, we investigate a related question by comparing the power of the $s$-colorful decision oracle with its random certificate counterpart. Our results show that, for every $1\le s<d$, access to uniformly random witness hyperedges yields an exponential separation in query complexity, whereas such separation disappears when $s=d$.

\paragraph{Conditional Sampling with Structural Constraints:} Furthermore, the $s$-colorful random certificate oracle may be viewed as a structured form of \emph{conditional sampling}, a powerful primitive that has led to exponential improvements for numerous testing and estimation problems over probability distributions
\cite{ChakrabortyFischerGoldhirshMatsliah/ITCS/2013/CondSamp,DBLP:journals/siamcomp/CanonneRS15/Conditional,CanonneRS14,BhattacharyyaC18,CanonneCKLW21,Bhattacharyya0P24}. Instead of conditioning on an arbitrary subset of the domain, our oracle conditions on the set of hyperedges intersecting a prescribed collection of pairwise disjoint vertex subsets. Theorem~\ref{thm:main-intro} demonstrates that this structured form of conditional sampling also leads to an exponential improvement over uniformly random hyperedge sampling. Finally, our work is closely related to the recent study of the \emph{power of random certificates} in structured query models \cite{RonTsur/TOCT/2016/PowerofAnExample,DBLP:journals/siamcomp/CanonneRS15/Conditional,ChakrabortyCKM/ICALP/2023/SATvsNPOracle}.


\paragraph{Fine-grained Complexity:} Finally, our results connect to
the fine-grained complexity framework of Dell and
Lapinskas~\cite{DBLP:journals/toct/DellL21} and Dell, Lapinskas, and
Meeks~\cite{DellLM22}. In this framework, the witnesses of a
computational problem (satisfying assignments, $d$-cliques, $d$-SUM
solutions) form the hyperedges of a $d$-uniform hypergraph, and a
decision procedure for restricted instances corresponds to a colorful
decision oracle; their central insight is that approximate counting
reduces to polylogarithmically many such decision calls when each of
the $d$ coordinates of a witness can be restricted to its own
prescribed set. We refine this picture along two axes. First, the
parameter $s$ measures the granularity of control: an $s$-colorful
oracle models a solver supporting only $s$ distinct restriction
classes, so that for $s<d$ distinct coordinates must share classes.
\Cref{thm:lower-intro} shows that this coarsening is severe, and not an
artifact of known techniques: for every $s<d$, approximate counting
requires polynomially many decision calls, so the reductions
of~\cite{DBLP:journals/toct/DellL21,DellLM22} make essential,
information-theoretic use of coordinate-wise restrictions. Second,
many natural solvers do not merely decide but \emph{search}: when a
witness exists, they return one. The $s$-colorful random certificate
oracle models such a solver whose returned witness is uniformly
distributed, and \Cref{thm:main-intro} shows that this search access
suffices for polylogarithmic-query counting for every $1\le s\le d$,
including the entire regime $s<d$ where decision access provably does
not. Thus, for approximate counting of size-$d$ witnesses, search with
uniformly random certificates is exponentially more powerful than
decision whenever the solver supports fewer than $d$ restriction
classes, a fine-grained analogue of the classical
decision-versus-search phenomenon. It is instructive to contrast this
with the result of~\cite{ChakrabortyCKM/ICALP/2023/SATvsNPOracle}
discussed above. In both settings, a decision oracle is upgraded to
return a uniformly random witness whenever one exists: the
$s$-colorful decision oracle $\cisq{}$ is upgraded to $\ciscq{}$ in
our setting, and the NP oracle is upgraded to a uniform-witness SAT
oracle in theirs. The outcomes are opposite: in our setting the
upgrade reduces the query complexity of approximate counting
exponentially for every $s<d$, whereas for approximate model counting
it yields no reduction at all. The difference lies in the strength of
the underlying decision oracle. For model counting, the decision
oracle is already as powerful as the search
oracle~\cite[Theorem~1.2]{ChakrabortyCKM/ICALP/2023/SATvsNPOracle},
leaving no gap for a witness to close, whereas for every $s<d$ the
$s$-colorful decision oracle is provably weak (\Cref{thm:lower-intro}),
and it is precisely this gap that uniformly random witnesses close.
The power of search over decision for approximate counting is thus
governed not by the witness distribution alone but by whether decision
access is already sufficient in the given query model.

\subsection{Notations}
We denote $[n]$ to denote the set of integers $\{ 1, 2, \ldots, n\}$. Given a vertex $\vertex$, we denote $\neighbour({\vertex})$ to be the set of neighbours of the vertex $\vertex$. Consequently, given a set of vertices $U \in \vertexset$, we denote $\neighbour{(U)} = \bigcup_{\vertex \in U} \neighbour(\vertex)$ to be the union of the neighbours of the vertices in $U$.

Throughout, we use $ \graph = ( \vertexset, \edgeset )$ to denote a $\hyperlen$-uniform hypergraph (or simply $\hyperlen$-hypergraph) where $ \vertexset $ is the set of $n$ vertices and $ \edgeset  \subseteq \binom{n}{d} $ is a family of $\edgecount$ hyperedges. For any subhypergraph $\sgraph$ of $\graph$, we denote its vertex set by $\vertexset(\sgraph)$, its set of hyperedges by $\edgeset(H)$, and its hyperedge count by $\edgecount(\sgraph)$.

\paragraph{$s$-partite $d$-hypergraph} Let $V_1, \ldots, V_\qsetno$ be mutually disjoint subsets of $\vertexset$. We define the induced $s$-partite subhypergraph on sets, denoted by $\graph \tbrac{V_1, \ldots, V_\qsetno}$ with vertex set $\cup_{i = 1}^s V_i$, hyperedge set is
$$ 
    E(G\tbrac{V_1, \ldots, V_s})  \left\{\edge \in \edgeset : \edge \cap V_{i} \neq \emptyset \text{ for all } i \in [\qsetno] \mbox{ and, }  e = \bigcup_{i\in [s]} (e \cap V_{i}) \right\}.
$$

\subsection{Organization of the paper}

We begin with Section~\ref{section: tech overview}, which provides a technical overview of our upper bound (Theorem~\ref{thm:main-intro}) and lower bound (Theorem~\ref{thm:lower-intro}). Section~\ref{section: s-par d-hyper} presents our algorithm for hyperedge estimation in an $s$-partite $d$-hypergraph. We then extend this algorithm to general $d$-hypergraphs in Section~\ref{section: full algorithm}. Section~\ref{Section: Lower Bounds} presents our lower bound, while Section~\ref{section: application} discusses applications of our results. Finally, Section~\ref{sec:conclusion_and_open_problems} concludes the paper and discusses open problems.

\section{Technical Overview}\label{section: tech overview}

\subsection{Upper Bounds}

We outline the proof of our main upper bound (Theorem~\ref{thm:main-intro}).
The algorithm has two layers: an outer reduction from general
$\hyperlen$-hypergraphs to $\qsetno$-partite instances, and an inner
estimator for $\qsetno$-partite $\hyperlen$-hypergraphs
(Theorem~\ref{theorem: s-par d-hyper}). The inner estimator combines three
ingredients: a collision-based estimator for hypergraphs with few hyperedges,
a sparsification step in which random certificates are used to \emph{verify}
that a randomly subsampled subhypergraph retains enough hyperedges, and a
telescoping-product argument that assembles the final estimate.

\paragraph{Why the existing graph techniques fail} For graphs ($\hyperlen = 2$), edge
estimation with global queries typically proceeds by decomposing the graph
into bipartite subgraphs: a uniformly random bipartition retains each edge
with probability exactly $1/2$ and, since the survival events of any two
distinct edges are pairwise independent, a Chebyshev bound shows that the
number of surviving edges concentrates around half the edge count. This is
the route taken in~\cite{harpeled18} and~\cite{CDGM26}. For
$\hyperlen$-hypergraphs, pairwise independence is lost: two hyperedges may
share up to $\hyperlen - 1$ vertices, so their survival events under a random
coloring or vertex subsampling are strongly correlated, and the second-moment
argument for a \emph{single} random decomposition breaks down. The known
hypergraph techniques~\cite{BhattacharyaBGM19, DellLM22} circumvent this
issue only in the fully colorful regime (\cisq{} with
$\qsetno = \hyperlen$), where a hyperedge can be pinned down coordinate by
coordinate, and do not extend to the case $\qsetno < \hyperlen$. Our
algorithm sidesteps the lack of independence altogether. The key observation
is that independence is only needed \emph{across} repetitions, not
\emph{within} a decomposition: instead of proving that one random
decomposition concentrates, we generate many independent candidate
decompositions and then use uniformly random certificates to \emph{test}
which candidate is good; since the candidate counts are bounded random
variables, Hoeffding's bound applies to their average with no assumptions on
within-decomposition correlations. This estimate-then-verify strategy is what allows a
single algorithm to handle all $1 \leq \qsetno \leq \hyperlen$.

\paragraph{Reduction to multipartite hypergraphs} Since a \ciscq{} query
conditions on $\qsetno$ pairwise disjoint vertex classes, we first reduce the
problem to $\qsetno$-partite hypergraphs. The algorithm uniformly and
independently partitions the vertex set $V$ into $\qsetno$ color classes
$V_1, \ldots, V_\qsetno$. A fixed hyperedge survives in the induced
$\qsetno$-partite subhypergraph if and only if its $\hyperlen$ vertices
receive all $\qsetno$ colors, which happens with probability
$f(\qsetno, \hyperlen)/\qsetno^{\hyperlen}$, where
$f(\qsetno, \hyperlen) = \sum_{i=0}^{\qsetno}(-1)^i\binom{\qsetno}{i}(\qsetno - i)^{\hyperlen}$
is the number of onto functions from $[\hyperlen]$ to $[\qsetno]$. Because of
the correlations described above, a single random partition need not
concentrate. However, each partition yields a hyperedge count in
$[0, \edgecount]$, and distinct partitions are mutually independent.
Repeating the partition
$t = O_{\qsetno, \hyperlen}\fbrac{\log(\delta^{-1})/\approxerror^{2}}$ times
therefore generates an ensemble of independent $\qsetno$-partite instances
such that, by a Hoeffding bound applied across the ensemble, the average
hyperedge count, rescaled by the survival probability, is a
$(1 \pm \approxerror)$-approximation of $\edgecount$ with probability
$1 - \delta$. It thus suffices to estimate the hyperedge count of each
instance up to accuracy $\Theta(\approxerror)$ with confidence
$1 - \Theta(1/t)$, and to take a union bound over the ensemble.

\paragraph{Small hyperedge count estimation via collision statistics} For a
$\qsetno$-partite $\hyperlen$-hypergraph $\sgraph$ with partite sets
$V_1, \ldots, V_\qsetno$, the query $\ciscq(V_1, \ldots, V_\qsetno)$ returns
a uniformly random hyperedge of $\sgraph$, so standard birthday-paradox
arguments~\cite{Mitzenmacher_Upfal_2005} apply. They yield two procedures: a tester that
distinguishes $\edgecount(\sgraph) \leq \threshold$ from
$\edgecount(\sgraph) \geq 8\threshold$ using
$O(\sqrt{\threshold}\log(\delta^{-1}))$ queries, and, whenever
$\edgecount(\sgraph) \leq \threshold$, an estimator that outputs a
$(1 \pm \approxerror)$-approximation of $\edgecount(\sgraph)$ using
$O(\sqrt{\threshold}\log(\delta^{-1})/\approxerror)$ queries. We invoke
these procedures only with $\threshold$ set to a large constant $\zeta$, so
each call costs $O(\log(\delta^{-1})/\approxerror)$ queries. The details can
be found in Section~\ref{section: tc} and Appendix~\ref{appendix: tc}.

\paragraph{Sparsification with verified selection} To reduce a hypergraph
$\sgraph$ with many hyperedges to the small-count regime, we sparsify: each
vertex of each partite set is kept independently with probability $1/2$, so
that a hyperedge survives with probability exactly $2^{-\hyperlen}$. Once
again, no concentration holds for a single trial. Instead, we draw
$t' = O(4^{\hyperlen}\log(\delta^{-1}))$ independent candidate
subhypergraphs $\sgraph^1, \ldots, \sgraph^{t'}$. Each count
$\edgecount(\sgraph^i)$ lies in $[0, \edgecount(\sgraph)]$ and the candidates
are mutually independent, so a Hoeffding bound shows that, with probability
$1 - \delta$, at least one candidate satisfies
$\edgecount(\sgraph^i) \geq \edgecount(\sgraph)/2^{\hyperlen + 1}$.

The random certificates now enter crucially: we must identify a good
candidate without being able to count. We draw
$O_d(\log(\delta^{-1}))$ random hyperedges from $\sgraph$ using
\ciscq{} queries and, for each $i \in [t']$, compute the fraction
$\approxfractionofedges(\sgraph^i)$ of sampled hyperedges that survive in
$\sgraph^i$. Since the certificates are uniformly distributed,
$\approxfractionofedges(\sgraph^i)$ is an unbiased estimator of the true
retained fraction $\edgecount(\sgraph^i)/\edgecount(\sgraph)$, and by a
Hoeffding bound every $\approxfractionofedges(\sgraph^i)$ deviates from its
mean by less than $1/2^{\hyperlen + 2}$ additively, simultaneously for all
$i$, with probability $1 - \delta$. Consequently, selecting any candidate
with $\approxfractionofedges(\sgraph^i) \geq 1/2^{\hyperlen + 2}$ is sound
in both directions: the good candidate guaranteed above passes the test, and
any candidate that passes satisfies
$
    \frac{\edgecount(\sgraph^i)}{\edgecount(\sgraph)} \geq \frac{1}{2^{\hyperlen + 3}}.
$
We call a subhypergraph satisfying this property an \emph{approximation
preserving subhypergraph}; the lower bound on the retained fraction is what
guarantees the quality of the final estimator. We emphasise that the decision
oracle \cisq{} offers no analogue of this verification step: it can certify
that a candidate contains \emph{some} hyperedge, but not that it retains a
constant \emph{fraction} of them. Indeed, as discussed in
Section~\ref{ssec:comparison_with_prior_work}, an oracle returning an \emph{arbitrary}
certificate can be simulated by $O_{\qsetno,\hyperlen}(\log n)$ adaptive
\cisq{} queries, so by Theorem~\ref{thm:lower-intro} it cannot support this
step either: uniformity of the certificate is exactly what is being used.

\paragraph{Constructing an approximation preserving sequence} The
construction proceeds via an iterative loop, starting with
$\sgraph_0 = \sgraph$. At each iteration $j$, we run threshold checking on
$\sgraph_j$ with the constant threshold $\zeta$: if the check reports
$\edgecount(\sgraph_j) \leq \zeta$, we halt and set
$\sgraph_\kappa = \sgraph_j$; otherwise, we sparsify $\sgraph_j$ and select
an approximation preserving subhypergraph $\sgraph_{j + 1}$ as above. Since
the number of iterations is itself random, we run iteration $j$ with
confidence parameter $\delta/(3 \cdot 2^{j})$; the resulting failure
probabilities form a geometric series, so a union bound over \emph{all}
iterations costs only $\delta$ in total, regardless of when the loop halts.
This yields a nested sequence of subhypergraphs with
$\vertexset(\sgraph) = \vertexset(\sgraph_0) \supseteq \ldots \supseteq \vertexset(\sgraph_\kappa)$
where, with probability $1 - \delta$, each $\sgraph_j$ satisfies
$
    \frac{\edgecount(\sgraph_j)}{\edgecount(\sgraph_{j - 1})} \geq \frac{1}{2^{\hyperlen + 3}}
$
and $\edgecount(\sgraph_\kappa) \leq 8\zeta = O(1)$.

To control $\kappa$, observe that every vertex survives each round
independently with probability $1/2$, so the round in which a fixed vertex is
removed is a geometric random variable, and $\kappa$ is at most the maximum
of $n$ such variables. Hence $\Pr[\kappa > k] \leq n \cdot 2^{-k}$, which
gives $\E[\kappa] = O(\log n)$, $\E[\kappa^2] = O(\log^2 n)$, and
$\E[\kappa^2 \log \kappa] = O(\log^2 n \cdot \log \log n)$; these are the
quantities that drive the query complexity below.

\paragraph{Telescoping estimation} For each $j \in [\kappa]$, define
$\fractionofedges_j = \frac{\edgecount(\sgraph_j)}{\edgecount(\sgraph_{j - 1})}$.
The total number of hyperedges follows a telescoping product:
\begin{equation*}
    \edgecount(\sgraph) = \frac{\edgecount(\sgraph_\kappa)}{\prod_{j = 1}^\kappa \fractionofedges_j}
\end{equation*}

Instead of estimating each $\edgecount(\sgraph_j)$ from scratch, we estimate
only the ratios $\fractionofedges_j$: we sample random hyperedges from
$\sgraph_{j-1}$ via \ciscq{} queries and record the fraction that survive in
$\sgraph_j$. Because
$\fractionofedges_j \geq 1/2^{\hyperlen + 3} = \Omega_\hyperlen(1)$ by
construction of the approximation preserving sequence, a multiplicative
$(1 \pm \approxerror_j)$-approximation of $\fractionofedges_j$ requires only
$O_d(\log(\delta^{-1})/\approxerror_j^{2})$ queries, independent
of $n$ and of $\edgecount(\sgraph_j)$. We take
$\approxerror_j = \approxerror/\Theta(\sqrt{\kappa})$ for every $j$, which
yields estimates satisfying $\widehat{\fractionofedges}_j = (1 + \varepsilon_j) \fractionofedges_j$ with the desired probability.~\footnote{A natural choice is to set $\varepsilon_j=\varepsilon/\Theta(\kappa)$. Then it yields a query
complexity of $\widetilde{O}(\log^3 n)$, motivating a more careful choice of
the parameter.}


\paragraph{Query complexity $\widetilde{O}(\log^2 n)$}
Within a single $\qsetno$-partite instance, the sequence construction costs
$O(\log(\delta^{-1}) + \kappa^2)$ queries and the birthday estimator
$O(1/\approxerror)$ queries; the dominant cost is ratio estimation. There are
$\kappa$ ratios, each estimated to accuracy
$\approxerror/\Theta(\sqrt{\kappa})$ with confidence $1/O(\kappa)$, at a
cost of $O_\hyperlen(\kappa \log \kappa/\approxerror^{2})$ queries per
ratio, hence $O_\hyperlen(\kappa^{2}\log\kappa/\approxerror^{2})$ in total.
Taking expectations over $\kappa$ using
$\E[\kappa^2 \log \kappa] = O(\log^2 n \cdot \log\log n)$ gives an expected
cost of $O_\hyperlen(\log^{2} n \cdot \log\log n/\approxerror^{2})$ per
instance. One logarithm is thus the length of the approximation preserving
sequence, and the other arises because each ratio must be estimated to
accuracy $\approxerror/\sqrt{\kappa}$. Both phenomena are governed solely by
the vertex-halving rate $1/2$, so the exponent of the $\log n$ factor is an
absolute constant: the parameters $\qsetno$ and $\hyperlen$ enter only
through the constants $2^{\hyperlen + 3}$, $4^{\hyperlen}$ and
$f(\qsetno, \hyperlen)$ hidden in $O_{\qsetno, \hyperlen}(\cdot)$. Finally,
the outer reduction runs the instance estimator
$t = O_{\qsetno,\hyperlen}(1/\approxerror^{2})$ times with confidence
parameter $\Theta(1/t)$, contributing a factor
$O(\log(1/\approxerror)/\approxerror^{2})$ and yielding the bound of
Theorem~\ref{thm:main-intro}. This is in sharp contrast to the
decision-oracle setting at $\qsetno = \hyperlen$, where the exponent of the
$\log n$ factor is $\Theta(\hyperlen)$~\cite{BhattacharyaBGM24, DellLM24}.

\subsection{Lower Bounds}
Let $\alpha:=\frac{d-s}{d-s+1}$. We prove the lower bound in two regimes based on the number of hyperedges. In the dense regime, where $m\ge c\,n^{d-1}/(2d)^{d-1}$, we prove a lower bound of $\widetilde{\Omega}((n^d/m)^\alpha)$ queries. In the sparse regime, where $m< c\,n^{d-1}/(2d)^{d-1}$, we prove a lower bound of $\widetilde{\Omega}(m^{\alpha/(d-1)})$ queries. These two bounds match, up to polylogarithmic factors, when $m=\Theta(n^{d-1})$. So the worst case query complexity is $\widetilde{\Omega}(n^\alpha)$, as stated in \Cref{thm:lower-intro}.

\subparagraph*{Dense regime.}
We prove the lower bound of $\widetilde\Omega((n^d/m)^\alpha)$ queries for $s$-colorful $\cisq$ oracle access, where $s<d$. The main idea is to hide a small set of vertices, each of which is part of many hyperedges. We show that finding even one such vertex takes many queries.

We prove the bound using Yao's min-max principle. We define two distributions, $\mathcal D_{\mathrm{Yes}}$ and $\mathcal D_{\mathrm{No}}$, over hypergraphs $G$. In both distributions, the vertex set $V(G)$ is split into $d$ parts $A_1,\ldots,A_d$ uniformly at random. Every colorful $d$ tuple $(a_1,\ldots,a_d) \in A_1\times\cdots\times A_d$ is made a hyperedge independently, with probability $\Theta(d^dm/n^d)$.

To build $G\sim\mathcal D_{\mathrm{No}}$, we add one more step. We form a set $B\subseteq A_1$ by placing each vertex of $A_1$ into $B$ independently, with probability $\Theta((2d)^dm\log n/n^d)$. For every vertex $v\in B$, we add all hyperedges of the form $\langle v,a_2,\ldots,a_d\rangle$ where $a_i\in A_i$ for each $i\ge 2$. Hence, $B$ forms the small set of high degree vertices to establish the indistinguishability argument.

We choose the constants above so that the following hold with probability $1-o(1)$. First, $|V(G)|=n$ and $|E(G)|=\Theta(m)$ under both distributions. Second, a graph from $\mathcal D_{\mathrm{Yes}}$ has at most $m/2$ hyperedges. Third, a graph from $\mathcal D_{\mathrm{No}}$ has $|B|=\Theta(m\log n/n^{d-1})$, and so has at least $m$ hyperedges. This means a graph from $\mathcal D_{\mathrm{No}}$ has at least twice as many hyperedges as a graph from $\mathcal D_{\mathrm{Yes}}$.

  Let $q := \widetilde{\Theta}\big((n^d/m)^{\alpha}\big)$. We want to show that
  any algorithm making $o(q)$ queries can not hit a vertex in the dense subset $B$. However, this can not hold if the query sets can be arbitrarily large. On the other hand, due to the construction of the graph distributions, any large query sets contain an edge with high probability. We leverage these observations by considering two possible cases for the queries. Given a $\cisq$ query $(Q_1,\ldots,Q_s)$, and a fixed threshold $\tau := \Theta((n^d/m)^{1/(d-s+1)})$, we consider the following cases: either $|Q_i|\le\tau$ for all $i$, or $|Q_i|>\tau$ for some $i$.

If $|Q_i|>\tau$ for some $i$, we show that the query sets either contain an edge or there exists a small certificate to show that they are not independent.  If each of the remaining query sets contain at least one colorful $(s-1)$-tuple (i.e. $(s-1)$ vertices of different colors consisting of one from each of the remaining query sets), then the value of $\tau$ ensures that these query sets contains at least $\Theta(n^d/m)$ possible colorful tuples. This is the reciprocal of the hyperedge probability. Hence, we can use a Chernoff argument to establish that any such query throughout the algorithm contains an edge with probability $1 - \smallo{1}$.

On the other hand, if there are a subset of the query sets that does not contain a colorful tuple, we show that all query sets in this subset contain less than $\tau$ vertices. Observe that the coloring $A_1,\ldots,A_d$ is chosen at random, and it does not depend on the query. So a large query set meets every color class with probability $1-o(1)$. By an argument using the Hall's marriage theorem, we show that any such set can not be part of such subsets. Based on these two results, we restrict the size of the query sets to $\tau$ for the remaining part of the proof.


We now bound the number of queries needed to tell $\mathcal D_{\mathrm{Yes}}$ apart from $\mathcal D_{\mathrm{No}}$ with probability at least $2/3$. Since $B=\emptyset$ under $\mathcal D_{\mathrm{Yes}}$, this task requires at least as many queries as are needed to touch a vertex of $B$ under $\mathcal D_{\mathrm{No}}$. Here, touching $B$ means the query's vertex set meets $B$. Since only small queries are useful, the chance that one such query touches $B$ is at most
\[
p = O\!\left(\tau\cdot\frac{m\log n}{n^d}\right)
= \widetilde O\!\left(\left(\frac{m}{n^d}\right)^\alpha\right).
\]
So the number of queries needed to touch $B$ is at least $1/p=\widetilde{\Omega}(q)$, which is $\widetilde\Omega((n^d/m)^\alpha)$. For this argument to work, we need $|B|=\Theta(m\log n/n^{d-1})$ to be $\Omega(\log n)$. This lets $|B|$ concentrate using Chernoff bounds. This is exactly why we need $m\ge c\,n^{d-1}/(2d)^{d-1}$.

\paragraph*{From non-adaptive to adaptive algorithms.}
The argument above only holds for non-adaptive algorithms. A
non-adaptive algorithm fixes all its queries in advance, before seeing
any answers. An adaptive algorithm can choose a query based on earlier
answers. For example, it might learn the colors of the vertices it has
already explored, and use them to grow a query set past $\tau$ while
avoiding hyperedges, specifically to search for $B$ more cheaply. 

To handle this case, we describe an \emph{augmented oracle} that simulates
the $\cisq{}$ oracle but reveals more information. At each stage, the
algorithm maintains a knowledge triplet $(K,\ell,e)$, where $K$ is the
set of explored vertices, $\ell$ labels each vertex of $K$ by the set
among $A_1,\ldots,A_d,B$ containing it, and $e$ records the hyperedges
among the vertices of $K$. Given a query $(Q_1,\ldots,Q_s)$, the
augmented oracle updates $(K,\ell,e)$ as follows. If some subset of
the query sets cannot form a colorful tuple, it reveals a minimal such
subset. If $|Q_i|\le\tau$ for all $i$, it reveals all queried vertices.
If $|Q_i|>\tau$ for some $i$, it samples a subset $R_i\subseteq Q_i$
of size $\tau$ uniformly at random and reveals these vertices; if
$(R_1,\ldots,R_s)$ is an independent set, the oracle fails and reveals
the entire graph. The answer to the original $\cisq{}$ query can always
be determined from the updated triplet. Hence, it is enough to prove
the lower bound for algorithms using the augmented oracle.

We view such an algorithm as a tree of depth $q$ where $q$ is the number of queries made by the algorithm. We say a node is
\emph{good} if no vertex of $B$ has been revealed so far. The key
observation is that, given the algorithm is at a good node, the
conditional distribution of the unexplored graph retains sufficient
randomness. The unexplored vertices are still assigned to
$A_1,\ldots,A_d$ and to $B$ independently at random. Moreover, every colorful $d$-tuple containing an unexplored vertex is still a
hyperedge independently with probability $\Theta(d^d m/n^d)$, exactly as in the construction. Hence, no matter how a
query was chosen based on earlier answers, the algorithm has not seen
the randomness that decides its answer. Using this, we show that each
query takes the algorithm from a good node to a bad node with
probability $o(1/q)$. If all query sets have size at most $\tau$, the
query hits $B$ with probability $o(1/q)$. A query
set larger than $\tau$ meets every color class with high probability,
so by Hall's marriage theorem it cannot be part of a minimal subset
that fails to form a colorful tuple. Finally, if some query set is
larger than $\tau$, the sampled sets $(R_1,\ldots,R_s)$ contain
$\Omega(\log^{d-s+1} n)$ potential hyperedges in expectation, so by a
Chernoff bound the oracle fails with probability $o(1/q)$. By a union bound over all
$q$ queries, the algorithm reaches a bad node with probability $o(1)$.

To complete the argument, we couple $\mathcal{D}_{\mathrm{Yes}}$ and
$\mathcal{D}_{\mathrm{No}}$: we draw a graph from
$\mathcal{D}_{\mathrm{Yes}}$, and obtain a graph from
$\mathcal{D}_{\mathrm{No}}$ by choosing $B\subseteq A_1$ and adding the
corresponding hyperedges. As long as no vertex of $B$ is revealed, the
algorithm behaves identically on both graphs. Hence, every good node is
reached with probability no larger under $\mathcal{D}_{\mathrm{No}}$
than under $\mathcal{D}_{\mathrm{Yes}}$, and the algorithm cannot
distinguish the two distributions with probability more than $o(1)$.
The formal analysis is presented in \Cref{Section: Lower Bounds}.

\subparagraph*{Sparse regime.}
The lower bound of $\widetilde\Omega(m^{\alpha/(d-1)})$, for the case $m< c\,n^{d-1}/(2d)^{d-1}$, follows by a reduction from the dense regime.

Let $n_2 := \big((2d)^{d-1}m/c\big)^{1/(d-1)}$, so that $m=c\,n_2^{d-1}/(2d)^{d-1}$. Let $H_2$ be an instance of the dense regime on $n_2$ vertices. We build $H_1 := H_2\sqcup H_0$, where $H_0$ is a rank $d$ hypergraph on $\Theta(n)$ new vertices, disjoint from $H_2$, respecting the same color partition, and containing no hyperedges. Then $H_1$ has $\Theta(n)$ vertices and the same number of hyperedges as $H_2$, matching the input size a sparse regime algorithm is assumed to handle. Since $m< c\,n^{d-1}/(2d)^{d-1}$, we get $n_2^{d-1}=(2d)^{d-1}m/c<n^{d-1}$, so $n_2<n$. Thus $(n_2,m)$ satisfies the requirement of the dense regime.

Every $s$-colorful $\cisq$ query on $H_1$ can be answered using exactly one query on $H_2$: restrict each queried set to its intersection with $V(H_2)$. The answer is unchanged, since $H_0$ has no hyperedges and none exist between $H_0$ and $H_2$.

Now suppose an algorithm violates the claimed sparse bound. Using the simulation above, it gives an algorithm on $H_2$ using $o((n_2^d/m)^\alpha)$ queries. Substituting $n_2=\big((2d)^{d-1}m/c\big)^{1/(d-1)}$, this is exactly the bound ruled out by the dense regime. Hence, a contradiction. This establishes the lower bound holds for the sparse regime.

\section{\texorpdfstring{Hyperedge Estimation of $\qsetno$-Partite $\hyperlen$-Hypergraph}{Hyperedge Estimation of s-Partite d-Hypergraph}}
In this section, we design a randomized algorithm to estimate the number of hyperedges in a $\qsetno$-partite $\hyperlen$-hypergraph. Specifically, we will show that,

\begin{restatable}[Hyperedge estimation of $\qsetno$-partite $\hyperlen$-hypergraph]{theorem}{spardhyper}\label{theorem: s-par d-hyper}
   For every integer $1\le s\le d$, $\approxerror \in (0, 1/3)$ and $\confidence \in (0, 1)$, there exists a randomized algorithm that, given access to the $s$-colorful random certificate oracle $\ciscq$ of a $\qsetno$-partite $\hyperlen$-hypergraph $\sgraph$, outputs a $(1\pm\varepsilon)$-approximation of the number of hyperedges of $\sgraph$ with probability at least $ 1 - \confidence$ using $O_{d}(\varepsilon^{-2}\log^2 n\log\log n\log (\confidence^{-1}))$ $\ciscq$ queries.
\end{restatable}

In Section~\ref{section: tc}, we present the threshold checking procedures. Section~\ref{section: sparsification} is organised into three parts that describe the procedure to generate an approximation preserving sequence. Section~\ref{section: approx preserve subhyper} introduces the sequential procedures \ensemblecolorpartite{} (Lemma~\ref{lemma: ensemblegraphpartite}) and \approxsubgraph{} (Lemma~\ref{lemma: approxsubhypergraph}) to generate an approximation preserving subhypergraph. Section~\ref{section: approx preserve sequence} presents the procedure \approxsubseq{} (Lemma~\ref{lemma: subsequence}) to generate an approximation preserving subsequence. Finally, in Section~\ref{section: s-par d-hyper} we describe the procedure to estimate the hyperedge count in a $s$-partite $d$-hypergraph.

\subsection{Threshold checking}\label{section: tc}

Let a fixed threshold ($ \threshold \geq 1$) be given. We consider the following threshold-checking problem: given a $\qsetno$-partite $\hyperlen$-hypergraph $ \sgraph $, determine whether $ \edgecount( \sgraph ) < \threshold $. If this condition holds, the goal is to estimate the $ (1 \pm \approxerror ) $ approximation of $ \edgecount( \sgraph )$.

To address this task, we present two algorithms \iscollision{} and \birthdayestimate{} based on the birthday paradox techniques \citep[Chapter 5]{Mitzenmacher_Upfal_2005}. \iscollision{} distinguishes between the cases $ \edgecount( \sgraph ) < \threshold $ or $ \edgecount( \sgraph ) \geq 8\threshold $ (Lemma \ref{lemma: collision detection}). \birthdayestimate{} outputs an $ (1 \pm \approxerror) $ approximation of $ \edgecount( \sgraph ) $, if $ \edgecount( \sgraph ) < \threshold $ (Lemma \ref{lemma: birthdayestimate}). The corresponding guarantees are stated below. For details, refer to Appendix~\ref{appendix: tc}.

\begin{restatable}[Correctness of \iscollision{}]{lemma}{isollison}
    For a $\qsetno$-partite $\hyperlen$-hypergraph $ \sgraph $  and a threshold $ \threshold \geq 1 $, the following statements hold for \iscollision{} with probability at least $ 1 - \confidence$:\label{lemma: collision detection}
\begin{enumerate}
    \item[(i)] If $ \edgecount( \sgraph ) \leq  \threshold $, then $\iscollision{}$ outputs $1$.
    \item[(ii)] If $ \edgecount( \sgraph ) \geq 8 \threshold $, then $\iscollision{}$ outputs $0$.
\end{enumerate}
The number of \ciscq{} queries made is $ O( \sqrt{ \threshold } \log ( \confidence^{ -1} ))$
\end{restatable}

Suppose that a collision is detected, i.e., \iscollision{} outputs $1$. We can obtain a
$ (1 \pm \approxerror)$-approximation of $\edgecount( \sgraph )$ using additional
\ciscq{} queries. This estimation procedure is implemented in \birthdayestimate{}.

\begin{restatable}[Correctness of \birthdayestimate{}]{lemma}{birth}\label{lemma: birthdayestimate}
For a $\qsetno$-partite $\hyperlen$-hypergraph $ \sgraph $, threshold $ \threshold \geq 1 $ and approximation error $ \approxerror > 0 $. If $ \edgecount( \sgraph ) \leq  \threshold $, then \birthdayestimate{} outputs an $ (1 \pm \approxerror) $ estimate of $ \edgecount( \sgraph ) $ with probability at least $1 - \confidence$. The number of \ciscq{} query made is $ \Theta\fbrac{  \approxerror^{-1} \sqrt{\threshold} \log \fbrac{ \confidence^{-1}} } $.
\end{restatable}

\subsection{Graph sparsification}\label{section: sparsification}
In this section, we will describe a procedure to generate an approximation preserving sequence (formally defined below) $\sgraph_0, \ldots, \sgraph_\kappa$ with probability at least $1-\confidence$. Moreover, the expected query complexity of the procedure is $O(\log^2 n)$.

\begin{definition}[Approximation preserving subhypergraph]
Let $\sgraph^1, \sgraph^2$ be subhypergraphs of $\sgraph$ such that  $\vertexset(\sgraph^2) \subseteq \vertexset(\sgraph^1)$, we say that $\sgraph^2$ is an \emph{approximation preserving subhypergraph} of $\sgraph^1$ if
\[
   \frac{m(\sgraph^2)}{m(\sgraph^1)}\ge \frac{1}{2^{\hyperlen + 3}}.
\]

\end{definition}

\begin{definition}[Approximation preserving sequence]
For a sequence of subhypergraphs $\sgraph_0, \sgraph_1, \ldots, \sgraph_\kappa$ with $ \vertexset(\sgraph_0) \supseteq \vertexset(\sgraph_1) \ldots \supseteq \vertexset(\sgraph_\kappa)$.
We say that $\sgraph_0, \sgraph_1, \ldots, \sgraph_\kappa$ is an \emph{approximation preserving sequence} if $\sgraph_{i-1}$ is an approximation preserving subset of $\sgraph_i$
for all $i \in [\kappa]$.
\end{definition}

Formally, we will prove the following lemma in this section. 
\begin{lemma}[Approximate sequence lemma]\label{lem:apx-seq} \label{lem: sequence}
There exists an algorithm that outputs an approximation preserving sequence  $\sgraph_0, \sgraph_1, \ldots, \sgraph_\kappa$  with probability at least $ 1 - \confidence$. The expected value of $\kappa$ is $O(\log \vertexcount)$, and the expected number of \ciscq{} queries made is $ O( \log^2 \vertexcount)$.
\end{lemma}




\subsubsection{Procedure to find an approximation preserving subhypergraph}\label{section: approx preserve subhyper}
In this section, we describe two procedures \ensemblecolorpartite{} and \approxsubgraph{} (see Algorithm \ref{alg: addestimation} and \ref{alg: ensemblecolorpartite}) for finding the approximation preserving subhypergraph $ \sgraph' $ of $\sgraph$, for any given $\qsetno$-partite $\hyperlen$-hypergraph $ \sgraph $. 

\ensemblecolorpartite{} takes a $\qsetno$-partite $\hyperlen$-hypergraph $\sgraph$ as input. It outputs multiple $\qsetno$-partite candidates for an approximation preserving subhypergraph. \approxsubgraph{} takes the original hypergraph $\sgraph$ and its subhypergraphs $\sgraph^1, \sgraph^2, \ldots, \sgraph^t$ (output of \ensemblecolorpartite{}) as input and generates an approximation preserving subhypergraph using the $\ciscq{}$ queries.  

These two procedures will be used in the next section for the construction of an approximation preserving sequence. Formally, we will prove the following:

\begin{lemma}[Correctness of \ensemblecolorpartite{}]\label{lemma: ensemblegraphpartite}
    For any $\qsetno$-partite $\hyperlen$-hypergraph $\sgraph$ with $\edgecount(\sgraph) \geq \thresholdconst$ and $(\varepsilon, \confidence) \in (0, 1)$. Let $\sgraph^1, \sgraph^2, \ldots, \sgraph^t$ be the output of \ensemblecolorpartite{}, then we have,
    \begin{equation*}
        \Pr \tbrac{ \exists i \in [t],\; \edgecount(\sgraph^i) \geq \frac{\edgecount(\sgraph)}{2^{\hyperlen + 1} }  } \geq 1 - \confidence
    \end{equation*}
    where $ t = O\fbrac{4^\hyperlen \log \fbrac{\confidence^{-1}}} $.
\end{lemma}

\begin{lemma}[Correctness of \approxsubgraph{}]\label{lemma: approxsubhypergraph}
    Let $\sgraph'$ be the output of \approxsubgraph{} that takes $\qsetno$-partite $\hyperlen$-hypergraph $\sgraph$ and its subhypergraphs $\sgraph^1, \ldots, \sgraph^\instances$ as input. Then, $\sgraph'$ is an approximation preserving subhypergraph of $\sgraph$ with probability at least $ 1 - \confidence$. The number of $\ciscq$ queries made is $ O\fbrac{ 4^\hyperlen \log (\confidence^{-1}) } $.
\end{lemma}

\begin{algorithm}
    \caption{\ensemblecolorpartite($ \sgraph, \confidence $)}\label{alg: ensemblecolorpartite}
    \begin{algorithmic}[1]
        \Require $\qsetno$-partite $\hyperlen$-hypergraph $\sgraph$ with partite sets $V_1, \ldots, V_\qsetno$, confidence parameter $\delta \in (0,1)$.
        \Ensure Subhypergraphs $\sgraph^1, \sgraph^2, \ldots, \sgraph^t$ of $\sgraph$ where $ t = O\fbrac{4^\hyperlen \log \fbrac{\confidence^{-1}}} $
        \For{$i = 1$ to $ t = O\fbrac{4^\hyperlen \log \fbrac{\confidence^{-1}}}$}
            \For{$j = 1$ to $\qsetno$}
                \State Sample each vertex of $V_j$ with probability $1/2$ and let $V_{ij}$ be the subset of sampled vertices.
            \EndFor
            \State Let $\sgraph^i = \sgraph\tbrac{V_{i1}, V_{i2}, \ldots, V_{i\qsetno}} $.
        \EndFor
        \State \Return $\sgraph^1, \sgraph^2, \ldots, \sgraph^t$
    \end{algorithmic}
\end{algorithm}

Now, we present the proof of Lemma~\ref{lemma: ensemblegraphpartite}

\begin{proof}[Proof of Lemma~\ref{lemma: ensemblegraphpartite}]
    Fix an $\qsetno$-partite $\hyperlen$-hypergraph $\sgraph^i$. For every hyperedge $\edge \in \edgeset(\sgraph) $, define a indicator random variable
    \begin{equation*}
    I_e = 
    \begin{cases}
        1 & \text{If $ e \in \sgraph^i $} \\
        0 & \text{Otherwise }
    \end{cases}
    \end{equation*}
    
    The expected number of hyperedges in $\sgraph^i$ is 
    \begin{equation*}
        \E{ [\edgecount(\sgraph^i)] } = \sum_{ e \in E(\sgraph) } \Pr[I_e = 1] = \frac{ \edgecount(\sgraph) }{ 2^\hyperlen }
    \end{equation*}
    
    Therefore, 
    \begin{equation*}
        \E{  \tbrac{\sum_{i = 1}^t \edgecount(\sgraph^i)} } = \frac{ \edgecount(\sgraph) \cdot t }{ 2^\hyperlen }
    \end{equation*}

    Now, applying Hoeffding's Bound (lemma~\ref{lemma: hoeffding}), with $t = \widetilde{O} \fbrac{ 4^\hyperlen }$ and $ 0 \leq \edgecount(\sgraph^i) \leq \edgecount(\sgraph) $, for all $ i \in [t]$ , we get
    \begin{align*}
        \Pr \tbrac{ \abs{ \frac{ 2^\hyperlen }{t} \cdot \sum_{i = 1}^t \edgecount(\sgraph^i) - \edgecount(\sgraph) } \geq \frac{ \edgecount(\sgraph) }{ 2 } } &= \Pr \tbrac{ \abs{  \sum_{i = 1}^t \edgecount(\sgraph^i) - \frac{t}{2^\hyperlen} \cdot \edgecount } \geq  \frac{t}{2^{\hyperlen + 1}} \cdot \edgecount(\sgraph)}\\
        &\leq 2 \cdot \exp\fbrac{ \frac{ 2 t }{ 2^{2\hyperlen + 1} } } \leq \delta 
    \end{align*}

    Therefore, we get
    \begin{equation*}
        \Pr \tbrac{  \frac{ 3\edgecount(\sgraph) }{ 2^{\hyperlen + 1} } \geq \frac{ \sum_{i = 1}^t \edgecount(\sgraph^i) }{ t } \geq \frac{ \edgecount(\sgraph) }{ 2^{\hyperlen + 1} } } \geq 1 - \confidence
    \end{equation*}

    which completes the proof.
\end{proof}

\begin{algorithm}
    \caption{\approxsubgraph($\sgraph, \sgraph^1, \ldots, \sgraph^\instances, \confidence $)}\label{alg: addestimation}
    \begin{algorithmic}[1]
        \Require Subhypergraphs $\sgraph^1, \ldots, \sgraph^\instances$ of $\sgraph$ where $t = O(4^d \log (\confidence^{-1}))$, \ciscq{} query access to $\sgraph$ and a confidence parameter $\delta \in (0,1)$.
        \Ensure A subhypergraph $\sgraph' $ of $\sgraph$.
        \State Uniformly and independently sample $ O( 4^\hyperlen \log (\confidence^{-1}) ) $ random hyperedges from $ \sgraph $ and let $ \samplededges $ be the multi-set of sampled hyperedges.
        \State Compute $ \approxfractionofedges ( \sgraph^i ) = \frac{ 1 }{ \size{ \samplededges } } \sum_{ \edge \in \samplededges} \indicator_{ \edge \in \edgeset(\sgraph^i) } $ for $ i \in [\instances] $. 
        \If{There exists $ i \in [\instances] $ such that $ \approxfractionofedges( \sgraph^i ) \geq \frac{ 1 }{ 2^{d + 2} } $}
            \State \Return $\sgraph^i$
        \Else 
            \State \Return Any $\sgraph^i$ arbitrarily.
        \EndIf
    \end{algorithmic}
\end{algorithm}

\begin{lemma}\label{lemma: addestimation}
    With probability at least $ 1 - \confidence$, the following statements hold for \approxsubgraph{} 
    \begin{enumerate}
    \item[(i)] If  $ \frac{ \edgecount( \sgraph^i ) }{ \edgecount(\sgraph) } \geq \frac{ 1 }{ 2^{d + 1} } $, then $ \approxfractionofedges( \sgraph^i ) \geq \frac{ 1 }{ 2^{d + 2} } $.%
    \item[(ii)] If  $ \frac{ \edgecount( \sgraph^i ) }{ \edgecount( \sgraph ) } < \frac{ 1 }{ 2^{d + 3} } $, then $ \approxfractionofedges( \sgraph^i ) < \frac{ 1 }{ 2^{d + 2} } $.
    \end{enumerate}
    The number of $\ciscq$ queries made is $ O\fbrac{ 4^\hyperlen \log (\confidence^{-1}) } $ samples.
\end{lemma}
\begin{proof}
    Fix a graph $ \sgraph^i $ for any $ i \in [t] $. For each sampled hyperedge $\edge \in \edgeset(\sgraph)$, define the indicator random variable
    \[
        X_j =
        \begin{cases}
        1 & \text{If $\edge \in \edgeset(\sgraph^i)$} \\
        0 & \text{otherwise.}
        \end{cases}
    \]
    Observe that $\E[X_j] = \frac{m(\sgraph^i)}{m(\sgraph)}$. 

    Define the fraction of sampled hyperedges that belong to $ \sgraph^i $ by
$
X = \frac{1}{\size{\samplededges}}
\sum_{j \in [\size{\samplededges}]} X_j.
$
where $ \samplededges $ is the multiset of sampled hyperedges obtained via \ciscq{} queries on $ \sgraph $. 
Note that $ X = \approxfractionofedges(\sgraph^i) $ and hence
\[
\E[\approxfractionofedges(\sgraph^i)]
= \frac{1}{\size{\samplededges}}
\sum_{j \in [\size{\samplededges}]}\E[X_j]
= \frac{\edgecount(\sgraph^i)}{\edgecount(\sgraph)}
\]

Now, applying Hoeffding's Bound (lemma~\ref{lemma: hoeffding}), with $\size{\samplededges}=O(4^\hyperlen  \log (\confidence^{-1}))$ we get
\[
    \Pr\!\left[ \approxfractionofedges(\sgraph^i) \geq
    \frac{m(\sgraph^i)}{m(\sgraph)}
    - \frac{1}{2^{\hyperlen + 2} }
    \right]
    \geq
    \exp\!\left(
    -\frac{\size{S}}{2^{2\hyperlen + 3} }
    \right) \geq 1 - \delta
\]

Using $ \frac{ \edgecount( \sgraph^i ) }{ \edgecount(\sgraph) } \geq \frac{ 1 }{ 2^{d + 1} } $, we get the result in $(i)$. Similarly, $(ii)$ can be proved.
\end{proof}

We are now ready to prove Lemma~\ref{lemma: approxsubhypergraph}.

\begin{proof}[Proof of Lemma~\ref{lemma: approxsubhypergraph}]
     Let $\mathcal{E}_{\text{aps}}$ denote the event that $\sgraph'$ is an approximation-preserving subhypergaph of $\sgraph$, that is, $ m(\sgraph')/m(\sgraph)\geq 1/2^{\hyperlen + 3}$.
     Let $\mathcal{E}$ be the event that there exists $i \in [t]$ such that $\edgecount( \sgraph^i ) / \edgecount( \sgraph ) \geq 1/2^{\hyperlen + 1}$. By Lemma~\ref{lemma: ensemblegraphpartite}, $\Pr[ \mathcal{E} ] \geq 1 - \confidence/2$.

     Let us condition on the event $\mathcal{E}$. It suffices to bound $\Pr[\mathcal{E}_{\text{aps}} \mid \mathcal{E}]$. Under the event $\mathcal{E}_{\text{aps}}$  there exists $i \in [t]$ such that $\edgecount( \sgraph^i ) / \edgecount( \sgraph ) \geq 1/2^{\hyperlen + 1}$, w.l.o.g assume $\edgecount( \sgraph^1 ) / \edgecount( \sgraph ) \ge 1/2^{\hyperlen + 1}$. If there exists $j \in [t] \setminus \{1\}$ such that $\edgecount( \sgraph^j) / \edgecount( \sgraph ) < 1/2^{\hyperlen + 3}$, then from Lemma~\ref{lemma: addestimation}, we have  $\approxfractionofedges(\sgraph^j) \leq 1/{2^{d + 2}}$ with probability atleast $1 - \confidence/2$ and \approxsubgraph{} will only output $$\sgraph^i \in \sbrac{ \sgraph^1 \cup \sgraph^j : \edgecount( \sgraph^j ) / \edgecount( \sgraph ) \geq 1/2^{\hyperlen + 3}}$$. Hence,
     \begin{equation*}
         \Pr[\mathcal{E}_{\text{aps}} \mid \mathcal{E}] \geq 1 - \frac{\confidence}{2}
     \end{equation*}
     Therefore,
     \begin{equation*}
         \Pr[\mathcal{E}_{\text{aps}}] = \Pr[\mathcal{E}] \cdot \Pr[\mathcal{E}_{\text{aps}} \mid \mathcal{E}] \geq 1 - \confidence
    \end{equation*}
\end{proof}
\subsubsection{Procedure to find approximation preserving sequence}\label{section: approx preserve sequence}
We now describe the algorithm \approxsubseq{} (see Algorithm~\ref{alg: subsequence}), which constructs an approximation preserving sequence used by the estimator.

\begin{algorithm}
    \caption{\approxsubseq($ \sgraph, \confidence $)}\label{alg: subsequence}
    \begin{algorithmic}[1]
        \Require \ciscq{} access to a $\qsetno$-partite $\hyperlen$-hypergraph $\sgraph$.
        \Ensure A nested sequence $\sgraph_1, \sgraph_2, \cdots, \sgraph_\iterations$ such that $V(\sgraph_1) \supseteq \ldots \supseteq V(\sgraph_\iterations)$
        \State Initialize $ \sgraph_0 \gets \sgraph $; $ \iscoll \gets 1 $; $ \iterations \gets 0 $  
        \While{$ true $}
            \State $\iscoll \gets \iscollision( \sgraph_\iterations, \thresholdconst, \frac{ \confidence }{ 3 \cdot 2^{ \iterations } } )$
            \If{$\iscoll$ equals $ 1 $}
                \State {}{break}
            \Else
                \State $\sgraph^{\iterations1}, \ldots, \sgraph^{\iterations t} \gets \ensemblecolorpartite( \sgraph_\iterations, \frac{ \confidence }{ 3 \cdot 2^{ \iterations } } )$\label{alg:subsequence line 7}
                \State $ \sgraph_{ \iterations + 1 } \gets \approxsubgraph( \sgraph_\iterations, \sgraph^{\iterations1}, \ldots, \sgraph^{\iterations t}, \frac{ \confidence }{ 3 \cdot 2^{ \iterations } } ) $
                \State $ \iterations \gets \iterations + 1 $
            \EndIf
        \EndWhile
        \State \Return $\sgraph_1, \sgraph_2, \cdots, \sgraph_\iterations$
    \end{algorithmic}
\end{algorithm}

\begin{restatable}[Correctness of \approxsubseq{}]{lemma}{}\label{lemma: subsequence}
Consider \approxsubseq{} as described in (Algorithm \ref{alg: subsequence}). With probability at least $ 1 - \confidence$, it outputs an approximation preserving sequence  $ \sgraph_1, \sgraph_2, \ldots, \sgraph_\iterations $ such that $\edgecount (\sgraph_{\kappa})\leq 8\thresholdconst$. Here $\thresholdconst$ is a large constant.  The number of $\ciscq$ queries made is $ O\fbrac{\log \confidence^{-1} + \iterations^2}$.
\end{restatable}

\begin{proof}
    {}{Correctness:}
    We can assume that $\edgecount(\sgraph_0)$ is sufficiently large, that is, $\edgecount(\sgraph) \geq \thresholdconst$. If not, then from Lemma~\ref {lemma: collision detection} \iscollision{} will output $1$ with probability $1 - \confidence$, and the algorithm will output nothing.
    
    Now, we prove the claim by induction on $\kappa$. For each $j \ge 1$, let $\event_j$ denote the event that the algorithm either terminates with $\edgecount(\sgraph_j) \leq 8 \thresholdconst$ or produces $\sgraph_{j + 1}$. Let $\event_{\le j}$ denote the event that $\event_i$ occurs for all $i \in [j]$, that is
    \[
        \event_{\le j} \;=\; \bigcap_{i=1}^{j} \mathcal{E}_i .
    \]

    First we claim that, If $\edgecount(\sgraph_j) \ge \thresholdconst$ then $ \Pr \tbrac{ \event_j \mid \event_{\le j - 1} } \ge 1 - \confidence_j$ where $\confidence_j =  \frac{ \confidence }{ 2^j }$. 

    We consider two cases $(\edgecount(\sgraph_j) > 8\thresholdconst)$ and $( \thresholdconst \leq \edgecount(U_j) \leq 8\thresholdconst )$, analyse them individually. In what follows, we assume that $\event_{\le j - 1}$ has occurred.
    
    \paragraph{Case $1$: $(\edgecount(\sgraph_j) > 8\thresholdconst)$.} Let $ \event_j = \event_1 \cap \event_2 \cap \event_3$ where $\event_1$ be the event that \iscollision{} output $0$, $\event_2$ be the event that \ensemblecolorpartite{} outputs $\sgraph_{j1}, \ldots, \sgraph_{j t}$ and $\event_3$ be the event that \approxsubgraph{} outputs $\sgraph_{j + 1}$. So, from Lemma~\ref{lemma: collision detection} we have $\Pr[\event_1 ] \ge \confidence_j/3$, from Lemma~\ref{lemma: ensemblegraphpartite} we have $\Pr[\event_2 \mid \event_1] \ge {\confidence_j}/{3} $ and from Lemma~\ref{lemma: approxsubhypergraph}, we have $\Pr[\event_3 \mid \event_2 \cap \event_1] \ge {\confidence_j}/{3} $. Then
    \begin{equation*}
        \Pr\tbrac{\event_j} = \Pr\tbrac{\event_1 \cap \event_2 \cap \event_3} \ge 1 - \confidence_j
    \end{equation*}

    \paragraph{Case $2$: $( \thresholdconst \leq \edgecount(U_j) \leq 8\thresholdconst )$.}
    In this case, we show that
    \[
    \Pr\!\left[\mathcal{E}_j \mid \mathcal{E}_{\le j-1} \right] \ge 1 - \delta_j
    \]
    regardless of whether \approxsubgraph{} terminates after executing \iscollision{} or proceeds to execute \ensemblecolorpartite{} and \approxsubgraph{}.
    
    In the former case, the claim follows directly from the definition of $\mathcal{E}_j$ together with the fact that $\edgecount(\sgraph_j) \le 8\thresholdconst$. In the latter case, \ensemblecolorpartite{} and \approxsubgraph{} produces $\sgraph_{j + 1}$, this occurs with probability $1 - 2\delta_j/3$. Hence
    \[
    \Pr\!\left[ \mathcal{E}_j \mid \mathcal{E}_{\le j-1} \right]
    \ge
    1 - \delta_j.
    \]

    Therefore, taking the union bound over all iterations, the probability that the algorithm produces a valid sequence is:
    \begin{equation*}
        \Pr \tbrac{\event_{\leq \iterations}} \ge 1 - \sum_{j = 1}^\iterations \confidence_j \ge 1 - \sum_{j = 1}^\infty \frac{ \confidence }{ 2^\iterations } = 1 - \confidence  
    \end{equation*}

    \paragraph{Query Complexity:} From Lemma~\ref{lemma: collision detection}, \ref{lemma: ensemblegraphpartite} and \ref{lemma: approxsubhypergraph}, the number of \ciscq{} queries made by \iscollision{} (for a constant threshold $ \threshold = \thresholdconst $), \ensemblecolorpartite{} and \approxsubgraph{} over $\iterations$ iterations is,
    \begin{align*}
        O\fbrac{\sum_{ j = 1 }^{\iterations} \log( \confidence^{-1} \cdot 2^\iterations ) } = O\fbrac{\log \confidence^{-1} + \iterations^2}
    \end{align*}
\end{proof}

Note that the number of iterations and the query complexity of  \approxsubseq{} (as mentioned in the above lemma) are random variables. We will argue their expectation in the following lemma. Moreover, 
when we discuss the algorithm of estimating the hyperedge count of $\qsetno$-partite $\hyperlen$-hypergraph in \ref{section: s-par d-hyper} (that uses \approxsubseq{} as a subroutine), we show that the query complexity of the this algorithm is $O(\kappa^2 \cdot \log \kappa)$, where $\kappa$ is the number of iterations of \approxsubseq{}. The following lemmas will be useful in establishing the query complexity in \ref{section: s-par d-hyper}.

\begin{lemma}\label{lemma: kappa}
$\E[\iterations]=O(\log \vertexcount)$, $\E[\iterations^2]=O(\log^2 \vertexcount)$, and $\E[\iterations^2 \log \iterations]=O(\log^2 \vertexcount \cdot \log \log \vertexcount)$.
\end{lemma}
\begin{proof}
    Let $\sgraph_1, \sgraph_2, \cdots, \sgraph_\iterations$ be the output of \approxsubgraph{}, where the vertices of $\sgraph_i$ is chosen independently with probability $1/2$ from the vertex set of $\sgraph_{i - 1}$, for all $i \in [\iterations]$ and $\sgraph_0 = \sgraph$. 

    Let $X_i$ be the number of rounds it takes for vertex $j$ in $\sgraph$ to be picked and removed. $X_j$ follows a geometric distribution where the probability of success of being removed on the $k$-th round is
    \begin{equation*}
        \Pr \tbrac{X_j = k} = \fbrac{ \frac{1}{2} }^{k - 1} \frac{1}{2} = \fbrac{ \frac{1}{2} }^{k}
    \end{equation*}
    Then the total number of rounds, $\iterations = \max\fbrac{X_1, X_2, \ldots, X_n}$. Therefore, 
    \begin{equation*}
        \Pr \tbrac{ \iterations > k} = \Pr \tbrac{ \max\fbrac{X_1, X_2, \ldots, X_n} > k} = \Pr \tbrac{ \bigcup_{ j = 1}^n \sbrac{ X_j > k} } \leq \sum_{ j = 1}^n \Pr \tbrac{ X_j > k} = n \fbrac{\frac{1}{2}}^k
    \end{equation*}

    Now, we will show that $\E[\iterations]=O(\log \vertexcount)$,  
    \begin{align*}
        \E \tbrac{\iterations} &= \sum_{k = 0}^\infty k \Pr \tbrac{ \iterations = k}\\
        &= \sum_{k = 0}^{\log n} k \Pr \tbrac{ \iterations = k} + \sum_{k = \log n + 1}^{\infty} k  \Pr \tbrac{ \iterations = k}\\
        &\leq \log n +  \sum_{\ell =  1}^{\infty} \fbrac{\log n + \ell}  \Pr \tbrac{ \iterations = \log n + \ell} \\
        &\leq \log n +  \sum_{\ell =  1}^{\infty}  \frac{n\fbrac{\log n + \ell}}{2^{\log n + \ell}}   \\
        &\leq \log n + O(1) = O(\log n)
    \end{align*}

    $\E[\iterations^2]=O(\log^2 \vertexcount)$, $\E[\iterations^2 \log \iterations]=O(\log^2 \vertexcount \cdot \log \log \vertexcount)$ can be shown in a similar manner.
\end{proof}

\subsection{\texorpdfstring{Procedure to hyperedge estimation of $s$-partite $d$-hypergraph}{Procedure to hyperedge estimation of s-partite d-hypergraph}}\label{section: s-par d-hyper}

In this section, we present the algorithm for estimating the number of hyperedges of $s$-partite $d$-hypergraph using \ciscq{} queries and prove Theorem~\ref{theorem: s-par d-hyper}.

\edgeestimationISC{} (Algorithm~\ref{alg: edgeestimationISC}) begins by calling \subsequence{} which outputs an approximation preserving sequence $\sgraph_1 , \ldots, \sgraph_\iterations $, where each $ \fractionofedges_i = \edgecount( \sgraph_i ) / \edgecount( \sgraph_{i-1} ) $ is bounded below by a constant. This ensures that each ratio can be estimated efficiently. To estimate each $ \fractionofedges_i $, we use the subroutine \mulestimation{} which returns an $ (1 \pm \approxerror)$-approximation of $ \fractionofedges_i$ . The procedure \mulestimation{} is described in \Cref{alg: mulestimation} and its guarantee is established in \Cref{lemma: mulestimation}.

\begin{algorithm}
    \caption{\mulestimation($ \sgraph', \sgraph, \approxerror, \confidence $)}\label{alg: mulestimation}
    \begin{algorithmic}[1]
        \Require A subhypergraph $ \sgraph' $ of $ \sgraph $ and \ciscq{} access to the hypergraph $\sgraph$.
        \Ensure An etimate of $\fractionofedges(\sgraph')$
        \State Uniformly and Independently sample $ O\fbrac{4^\hyperlen \log ( \confidence^{-1} ) / \approxerror^2 } $ random hyperedges from $ \sgraph $ and let $ \samplededges $ be the multi-set of sampled hyperedges.
        \State \Return $ \approxfractionofedges ( \sgraph' ) = \frac{ 1 }{ \size{ \samplededges } } \sum_{ \edge \in \samplededges} \indicator_{ \edge \in \edgeset(\sgraph') } $. 
    \end{algorithmic}
\end{algorithm}

\begin{restatable}[Correctness of \mulestimation{}]{lemma}{mulest}\label{lemma: mulestimation}
    For hypergraph $ \sgraph' $ such that $ \vertexset(\sgraph') \subseteq \vertexset(\sgraph) $ and $ \approxerror > 0$ , If $ \fractionofedges( \sgraph' ) = \frac{ \edgecount( \sgraph' ) }{ \edgecount( \sgraph ) } \geq \frac{ 1 }{ 2^{d + 3} } $, then \mulestimation{} outputs $ \approxfractionofedges ( \sgraph' ) = ( 1 + x ) \fractionofedges( \sgraph ) $, where $\E[x] = 0 $ and $x \in (-\approxerror, \approxerror)$ with probability at least $ 1 - \confidence $. The number of \ciscq{} queries made is $ O\fbrac{ 4^d \log ( \confidence^{-1} )  / \approxerror^2 } $.
\end{restatable}

\begin{proof}
    Let $ X_j$ be the indicator random variable such that,
\begin{equation*}
    X_j = 
    \begin{cases}
        1 & \text{If the $ j^{th} $ sampled hyperedge is in $ \sgraph' $} \\
        0 & \text{Otherwise }
    \end{cases}
\end{equation*}
Define the number of hyperedges sampled that belong to $ \sgraph' $ as $ X  = \sum_{ j \in [\size{ \samplededges }] } X_j$, where $ \samplededges $ is the multi-set of sampled hyperedges and the expected size of $ X $ is
\begin{equation*}
    \E\tbrac{ X } =  \sum_{ j \in [\size{S}] } \E\tbrac{ X_j } = \size{ S } \cdot \frac{ \edgecount{ (\sgraph') }}{ \edgecount{ (\sgraph) } } = \size{ S } \cdot \fractionofedges( \sgraph' ) \geq \frac{ \size{ \samplededges } }{ 2^{d + 3} }
\end{equation*}
Therefore,
\begin{align*}
    \E[x] = \E\tbrac{ \frac{\approxfractionofedges ( \sgraph' )}{\fractionofedges( \sgraph )}} - 1 = \frac{\E[X]}{\size{ S } \fractionofedges( \sgraph' )} -1 = 0
\end{align*}
By the Chernoff Bound (see Lemma~\ref{lemma: Multiplicative Chernoff Bound}), we have 
\[
    \Pr[x \in (-\approxerror, \approxerror)] = \Pr\tbrac{ \approxfractionofedges(\sgraph') \in (1\pm\approxerror)\fractionofedges(\sgraph') } = \Pr\tbrac{ X \in (1\pm\approxerror)\E[X] } \ge 1 - 2\exp\fbrac{ - \frac{\approxerror^2 \size{\samplededges}}{96} } \ge 1 - \confidence
\]
where the last step uses $\size{\samplededges}=O\left( 4^d \log (\delta^{-1})/\varepsilon^2\right)$.
\end{proof}

\begin{algorithm}
    \caption{\edgeestimationISC($ \sgraph, \approxerror, \confidence $)}\label{alg: edgeestimationISC}
    \begin{algorithmic}[1]
        \Require \ciscq{} access to the hypergraph $ \sgraph $.
        \Ensure An estimate of $ \approxedgecount $
        \For{$ j = 1 $ to $ s = O(\log \confidence^{-1}) $}\label{edgeestimationISC: line1}
        
            \State $ \sgraph_1, \sgraph_2, \cdots, \sgraph_\iterations \gets \approxsubseq(\sgraph, \frac{1}{5})$
            \State $ \approxedgecount_b \gets \birthdayestimate( \sgraph_\iterations, 8\thresholdconst, \frac{\approxerror }{ 3 }, \frac{1}{16} )$\label{edgeestimation: line 2}
            \State Initialize $ \approxfractionofedges_i \gets 1 $, for all $ i \in [\iterations]$
            \For{$ i = 1 $ to $ \iterations $}
                \State  $ \approxfractionofedges_i \gets \mulestimation( \sgraph_i, \sgraph_{i - 1}, \frac{ \approxerror }{ 200 \sqrt{\iterations} }, \frac{ 1 }{ 16 \iterations })$ \label{edgeestimation: line 5}
            \EndFor
            \State $ \approxedgecount_j \gets \approxedgecount_b / { \prod_{ i = 1 }^{ \iterations } \approxfractionofedges_i }$
        \EndFor
        \State \Return $ \mathrm{median} \fbrac{ \approxedgecount_1,  \approxedgecount_2, \ldots, \approxedgecount_s} $
    \end{algorithmic}
\end{algorithm}

We prove Theorem~\ref{theorem: s-par d-hyper} by proving Lemma~\ref{lem:final:query} and Lemma~\ref{lem:final-correct}.

\begin{lemma}\label{lem:final:query}
    The expected query complexity of \edgeestimationISC{} is $O_{\hyperlen}(\log^2 n \cdot \log \log n \cdot \log \confidence^{-1}/ \approxerror^2)$.
\end{lemma}
\begin{proof}
    From Lemma~\ref{lemma: mulestimation}, the number of \ciscq{} queries made by \mulestimation{} in every iteration is
    \begin{align*}
        O_{\hyperlen}\fbrac{ \frac{\log(\iterations )}{ \approxerror'^2 } } \text{where } \approxerror' = \frac{ \approxerror }{ 200 \sqrt{\iterations} }
    \end{align*}
    Summing over $ \iterations $ iterations, we get
    \begin{align*}
        O_{\hyperlen}\fbrac{ \sum_{ j = 1 }^{ \iterations } \frac{ \log \iterations }{ \approxerror'^2 } } = O_{\hyperlen}\fbrac{\iterations^2 \frac{ \log \iterations }{ \approxerror^2 } }
    \end{align*}
    From Lemma~\ref{lemma: subsequence}, the number of \iscq{} queries made by \approxsubseq is $ O( \iterations^2 ) $. Finally, from Lemma~\ref{lemma: birthdayestimate}, the number of \ciscq{} queries made by \birthdayestimate{} is $O\fbrac{ { \sqrt{8\thresholdconst} }/{ \approxerror}}= O(1/\approxerror)$. Combining all contributions, the query complexity of a single execution of the outer loop is dominated by \mulestimation{}, and hence is $O_{\hyperlen}(\iterations^2 \log \iterations / \approxerror^2)$. From Lemma \ref{lemma: kappa}, 
    $$
        \E[\iterations^2 \log \iterations] = O_{\hyperlen}(\log^2 \vertexcount \cdot \log \log \vertexcount)
    $$ 
    Since the outer loop is repeated $O(\log \confidence^{-1})$ times, the expected query complexity  is
    $$
        O_{\hyperlen}(\log^2 n \cdot \log \log n \cdot \log \confidence^{-1}/ \approxerror^2)
    $$
    as claimed.
\end{proof}

\begin{lemma}\label{lem:final-correct}
    \edgeestimationISC{}, outputs a $(1 \pm \approxerror)$-approximation to $m$ with probability at least $ 1 - \confidence$.
\end{lemma}

\begin{proof} 
    In the outer loop of \edgeestimationISC{} (Line~\ref{edgeestimationISC: line1}), the algorithm is repeated $O(\log \confidence^{-1})$ times to boost its success probability to at least $1 - \confidence$ using the standard median-of-mean tricks. In what follows, we prove that a single execution of \edgeestimationISC{} outputs a $(1 \pm \approxerror)$-approximation to $m$ with probability at least $2/3$.
    
    Let $\mathcal{E}_1 = \subseqevent \cap \birthdayevent$, where $\subseqevent$ be the event that \approxsubseq{} outputs an approximation-preserving sequence $\sgraph_1, \sgraph_2, \ldots, \sgraph\iterations$ with $ \edgecount( \sgraph_\iterations ) \leq 8 \thresholdconst$. From lemma \ref{lemma: subsequence} we have $\Pr[\subseqevent] \geq 4/5$.
    
    $ \birthdayevent $ be the event that \birthdayestimate{} outputs an $(1 \pm \approxerror/3)$ approximation of $ \edgecount( \sgraph_\iterations ) $, from Lemma \ref{lemma: birthdayestimate}, we have $\Pr[\birthdayevent \mid \subseqevent] \geq 15/16$. Therefore,
    \begin{equation*}
        \Pr[\mathcal{E}_1] = \Pr[\subseqevent]\Pr[\birthdayevent \mid \subseqevent] \geq 1 - \fbrac{ \frac{1}{5} + \frac{1}{16} } \geq \frac{59}{80}
    \end{equation*}

    Let $\mathcal{E}_2$ be the event that \mulestimation$\left( \sgraph_i, \sgraph_{i - 1}, \varepsilon', { 1 }/{ 16 \iterations}\right)$ outputs  $\widehat{p_i}=(1 \pm x_i/4)p_i$ such that $\E[x_i]=0$ and $x_i \in \fbrac{ -4{\approxerror'}, 4\approxerror'}$ for all $i \in [\iterations]$, where $ \approxerror' = \frac{ \approxerror }{ 200 \sqrt{ \iterations }}$ . Therefore, using Lemma~\ref{lemma: mulestimation} and applying the union bound over all $i \in [\iterations]$, 
    \begin{equation*}
        \Pr[\mathcal{E}_2 \mid \mathcal{E}_1] = 1 - \sum_{ j = 1 }^{ \iterations } \frac{ 1 }{ 16 \iterations } \geq \frac{15}{16}.
    \end{equation*}
    
    Let $\mathcal{E}_3$ be the event that $\abs{\sum_{i = 1 }^\iterations x_i} \leq \approxerror/9$. 
    By Hoeffding's bound (Lemma~\ref{lemma: hoeffding}),
    \begin{align}
        \Pr\!\left[
        \mathcal{E}_3 \mid \mathcal{E}_1 \cap \mathcal{E}_2
        \right]
        &=
        1-\Pr\!\left[
        \left|
        \sum_{i=1}^{\iterations} x_i
        \right|
        >
        \frac{\approxerror}{9}
        \,\middle|\,
        \mathcal{E}_1 \cap \mathcal{E}_2
        \right]
        \nonumber \\
        &\ge 1-
        2\exp\!\left(
        -
        \frac{
        2\approxerror^2
        }{
        81 \sum_{i=1}^{\iterations}
        \left(
        \frac{2\approxerror}{50\sqrt{\iterations}}
        \right)^2
        }
        \right) \nonumber \\
        &\ge 1-\frac{1}{e^{12}}. \label{ineq: hoeff}
    \end{align}
    
    Hence,
     \begin{align*}
        \Pr[ \mathcal{E}_1 \cap \mathcal{E}_2 \cap \mathcal{E}_3] &= \Pr[\mathcal{E}_1] \cdot \Pr[\mathcal{E}_2 | \mathcal{E}_1] \cdot \Pr[\mathcal{E}_3 \mid \mathcal{E}_1 \cap \mathcal{E}_2] \geq 1 - \fbrac{ \frac{21}{80} + \frac{ 1 }{ 16 } + \frac{ 1 }{ e^{12} } } \geq \frac{2}{3}
    \end{align*}
    In what follows, we assume that all of the events $\mathcal{E}_1$, $\mathcal{E}_2$ and $\mathcal{E}_3$ occur. Since, $\Pr[ \mathcal{E}_1 \cap \mathcal{E}_2 \cap \mathcal{E}_3] \geq 2/3$,  it is enough to argue that \edgeestimationISC{} returns an $(1 \pm \approxerror)$-approximation to $m$ when all $\mathcal{E}_1$, $\mathcal{E}_2$ and $\mathcal{E}_3$ holds.

    Combining the output of \mulestimation{} over the $ \iterations $ iterations, we get,
    \begin{align*} 
          \prod_{i = 1}^{ \iterations } \approxfractionofedges_i = \prod_{i = 1 }^{ \iterations } \fractionofedges_i \cdot \prod_{i = 1}^{ \iterations } \fbrac{ 1 + \frac{x_i}{4} }
    \end{align*}
     Using the fact that for each $x_i \in (0,1)$, it can be shown that 
    \begin{align*}
        1 - \frac{ \sum_{i = 1}^{ \iterations } x_i }{ 4 }  \leq \prod_{i = 1}^{ \iterations } \fbrac{ 1 + \frac{x_i}{4} } \leq 1 + \sum_{i = 1}^{ \iterations } x_i
    \end{align*}

    \noindent Since event $\mathcal{E}_3$ occurs,
    \begin{align*}
        \prod_{i = 1 }^{ \iterations } \fractionofedges_i \cdot \fbrac{1 - \frac{\approxerror}{36}} \leq \prod_{i = 1}^{ \iterations } \approxfractionofedges_i \leq \prod_{i = 1 }^{ \iterations } \fractionofedges_i \cdot \fbrac{1 + \frac{\approxerror}{9}}
    \end{align*}
    Using the fact that for any $\approxerror \in (0, 1/3)$, we have $1 - 3\approxerror \leq (1 + \approxerror)^{-1}$ and $ (1 - \approxerror)^{-1} \leq 1 + 3\approxerror $. We get, 
    \begin{align}\label{ineq: xyz}
        \fbrac{1 - \frac{\approxerror}{3}} \cdot \frac{1}{\prod_{i = 1 }^{ \iterations } \fractionofedges_i}
         \leq \frac{1}{\prod_{i = 1}^{ \iterations } \approxfractionofedges_i} \leq \frac{1}{\prod_{i = 1 }^{ \iterations } \fractionofedges_i} \cdot \fbrac{1 + \frac{\approxerror}{12}}
    \end{align}

    Since $\mathcal{E}_1$ occurs, i.e., $\birthdayevent$ occurs, $ \approxedgecount_b$  satisfies the following:
    \begin{align} \label{ineq: birthdayest}
         \fbrac{ 1 - \frac{\approxerror}{3} } \cdot \edgecount( \sgraph_\iterations )  \leq \approxedgecount_b \leq \fbrac{ 1 + \frac{\approxerror}{3} } \cdot \edgecount( \sgraph_\iterations )
    \end{align}
    Combining~\eqref{ineq: xyz} and \eqref{ineq: birthdayest}, we get,
    \begin{align*}
        \fbrac{ 1 - \frac{\approxerror}{3} }^2 \cdot \frac{ \edgecount( \sgraph_\iterations ) }{ \prod_{i = 1 }^{ \iterations } \fractionofedges_i} 
        \leq 
        \widehat{m}=\frac{ \approxedgecount_b }{ \prod_{i = 1}^{ \iterations } \approxfractionofedges_i} 
        \leq
        \fbrac{ 1 + \frac{\approxerror}{12}}\fbrac{ 1 + \frac{\approxerror}{3}} \cdot \frac{ \edgecount( \sgraph_\iterations ) }{ \prod_{i = 1}^{ \iterations } \fractionofedges_i}
    \end{align*} 
    
    Since, $ \fractionofedges_i = \frac{ \edgecount( \sgraph_i) }{ \edgecount( \sgraph_{i - 1} )}$ for each $i \in [\kappa]$ and $ \sgraph_0 = G $, $ \frac{ \edgecount( \sgraph_\iterations ) }{ \prod_{i = 1 }^{ \iterations } \fractionofedges_i} = \edgecount( \sgraph_0 ) = \edgecount $. So, simplifying the above expression, we have 
    \begin{align*}
        \fbrac{ 1 - \approxerror } \edgecount \leq \approxedgecount \leq \fbrac{ 1 + \approxerror} \edgecount 
    \end{align*}
\end{proof}

We are now ready to prove \Cref{theorem: s-par d-hyper}.

\begin{proof}[Proof of \Cref{theorem: s-par d-hyper}]
The correctness of the algorithm follows from Lemma~\ref{lem:final-correct}, and its query complexity follows from Lemma~\ref{lem:final:query}.
\end{proof}

\section{\texorpdfstring{Hyperedge Estimation of $d$-hypergraph}{Hyperedge Estimation of d-hypergraph}}\label{section: full algorithm}
In this section, we present the full algorithm for estimating the number of hyperedges of $\hyperlen$-hypergraph using \ciscq{} queries and prove our main upper bound stated below.

\begingroup
\renewcommand{\thetheorem}{1}
\begin{theorem}\label{theorem: main result}
    For every integer $1\le s\le d$ and $\approxerror \in (0, 1/3)$, there exists a randomized algorithm that, given access to the $s$-colorful random certificate oracle $\ciscq$, outputs a $(1\pm\varepsilon)$-approximation of the number of hyperedges with probability at least $2/3$ using $O_{s,d} \fbrac{\varepsilon^{-4}\log^2 n\log\log n\log^2(1/\varepsilon)}$ $s$-colorful random certificate queries.
\end{theorem}
\endgroup

\edgeestimation{} (Algorithm~\ref{alg: edgeesimation}) begins by calling \ensemblecolor{}  which outputs multiple $\qsetno$-partite $\hyperlen$-hypergraphs by uniformly partition the vertex set into $\qsetno$ classes. The procedure for \ensemblecolor{} is described in Algorithm~\ref{alg: ensemblecolor} and its guarantee is established in Lemma~\ref{lemma: ensemblecolor}.

\begin{algorithm}
    \caption{\ensemblecolor($ \graph, \approxerror, \confidence $)}\label{alg: ensemblecolor}
    \begin{algorithmic}[1]
        \Require \ciscq{} query access to $G$, error parameter $\approxerror \in (0, 1) $ and a confidence parameter $\delta \in (0,1)$.
        \For{$i = 1$ to $t = O\fbrac{f^2(\qsetno, \hyperlen)\approxerror^{-2} \log \fbrac{\confidence^{-1}}}$}
            \State Uniformly and independently partition $ V $ into $ V_1, V_2, \ldots, V_\qsetno $ and let $\graph_i = \graph\tbrac{V_1, V_2, \ldots, V_\qsetno} $.\label{ensemblecolor: line 1}
        \EndFor
        \State \Return $\graph_1, \graph_2, \ldots, \graph_t$
    \end{algorithmic}
\end{algorithm}

\begin{lemma}\label{lemma: ensemblecolor}
    Let $\onto{ \qsetno }{\hyperlen} = s^d / \sum_{i = 0}^\qsetno (-1)^i \binom{\qsetno}{i}(\qsetno - i)^\hyperlen$, then \ensemblecolor{} outputs a set of $\qsetno$-partite $\hyperlen$-hypergraph $\graph_1, \graph_2, \ldots, \graph_t$ where $ t = \widetilde{O} \fbrac{ \frac{ f^2(\qsetno, \hyperlen) }{ \approxerror^2 }} $, such that we have,
    \begin{equation*}
        \Pr \tbrac{ \frac{\onto{\qsetno }{\hyperlen}}{t} \cdot \sum_{i = 1}^t \edgecount(\graph_i) \in \fbrac{1 \pm \approxerror} \edgecount  } \geq 1 - \confidence
    \end{equation*}
\end{lemma}
\begin{proof}
    Fix an $\qsetno$-hypergraph $G_i$. For every hyperedge $\edge \in \edgeset(\graph) $, define a indicator random variable
    \begin{equation*}
    I_e = 
    \begin{cases}
        1 & \text{If $ e \in \graph_i $} \\
        0 & \text{Otherwise }
    \end{cases}
    \end{equation*}
    
    The expected number of hyperedges in $G_i$ is 
    \begin{equation*}
        \E{ [\edgecount(G_i)] } = \sum_{ e \in E(G) } \Pr[I_e = 1] = \frac{ \edgecount }{ \onto{\qsetno }{\hyperlen} }
    \end{equation*}
    
    Therefore, 
    \begin{equation*}
        \E{  \tbrac{\sum_{i = 1}^t \edgecount(G_i)} } = \frac{ \edgecount \cdot t }{ \onto{\qsetno }{\hyperlen} }
    \end{equation*}

    Now, applying Hoeffding's Bound~\ref{lemma: hoeffding}, with $t = \widetilde{O} \fbrac{ \frac{ f^2(\qsetno, \hyperlen) }{ \approxerror^2 }}$ and $ 0 \leq \edgecount(\graph_i) \leq \edgecount $, for all $ i \in [t]$ , we get
    \begin{align*}
        \Pr \tbrac{ \abs{ \frac{\onto{\qsetno }{\hyperlen}}{t} \cdot \sum_{i = 1}^t \edgecount(\graph_i) - \edgecount } \geq \approxerror \edgecount } &= \Pr \tbrac{ \abs{  \sum_{i = 1}^t \edgecount(\graph_i) - \frac{t}{\onto{\qsetno }{\hyperlen}} \cdot \edgecount } \geq  \frac{\approxerror t}{\onto{\qsetno }{\hyperlen}} \cdot \edgecount }\\
        &\leq 2 \cdot \exp\fbrac{ \frac{ 2 t \approxerror^2 }{ f^2(\qsetno, \hyperlen) } } \leq \delta 
    \end{align*}
\end{proof}

\begin{algorithm}
    \caption{\edgeestimation($ \graph, \approxerror $)}\label{alg: edgeesimation}
    \begin{algorithmic}[1]
        \Require \ciscq{} access to the hypergraph $ \graph $.
        \Ensure An estimate of $ \approxedgecount $
        
        \State $ G_1, G_2, \cdots, G_t \gets \ensemblecolor \left( \graph, \frac{\approxerror}{3}, \frac{1}{6} \right)$\label{edgeestimation: line 1}

        \For{$ j = 1 $ to $ t $}
            \State  $ \approxedgecount(G_j) \gets \edgeestimationISC\fbrac{\graph_j, \frac{\approxerror}{3}, \frac{1}{6t} }$
        \EndFor
        
        \State \Return $ \approxedgecount \gets \frac{\onto{\qsetno }{\hyperlen}}{t} \cdot \sum_{i = 1}^t \approxedgecount(\graph_i)$
        
    \end{algorithmic}
\end{algorithm}

Now, we prove the main upper bound. 

\begin{proof}[Proof of Theorem~\ref{theorem: main result}]
    \noindent\textbf{Correctness}: Let $\mathcal{E}_1 $ be the event that \ensemblecolor{} outputs a set of $\qsetno$-partite $\hyperlen$-hypergraph $ G_1, G_2, \cdots, G_t$ such that 
    \begin{equation*}
        \frac{\onto{\qsetno }{\hyperlen}}{t} \cdot \sum_{i = 1}^t \edgecount(\graph_i) \in \fbrac{1 \pm \frac{\approxerror}{3}} \edgecount 
    \end{equation*}
    Then, from Lemma~\ref{lemma: ensemblecolor}, $\Pr \tbrac{ \mathcal{E}_1 } \geq 5/6$.

    Let, for any $i \in [t]$, $\mathcal{E}_{2i}$ be the event that \edgeestimationISC{} outputs an $(1 \pm \approxerror/3)$-approximation of $\edgecount(G_i)$. By Lemma~\ref{lem:final-correct}, $\Pr \tbrac{ \mathcal{E}_{2i} \mid  \mathcal{E}_1 } \ge 1 - 1/6t$. Let $\mathcal{E}_2$ be the event that $\mathcal{E}_{2i}$ occurs for all $i \in [t]$, that is 
    \begin{equation*}
        \mathcal{E}_2 = \bigcap_{ i = 1 }^t \mathcal{E}_{2i}
    \end{equation*}
    Using the union bound, we can write $\Pr \tbrac{ \mathcal{E}_{2} \mid  \mathcal{E}_1 } \ge 5/6$ which implies
    \begin{equation*}
        \Pr \tbrac{ \left. \sum_{i = 1}^t \approxedgecount(G_i) \in \fbrac{1 \pm \frac{\approxerror}{3}} \sum_{i = 1}^t \edgecount(G_i)  \,\middle|\, \mathcal{E}_1 \right. } \geq \frac{5}{6}
    \end{equation*}
    Therefore, $\Pr \tbrac{ \mathcal{E}_1 \cap \mathcal{E}_2} = \Pr[ \mathcal{E}_1] \Pr[ \mathcal{E}_2 \mid \mathcal{E}_1] \geq 2/3$.

    In what follows, we assume that events $\mathcal{E}_1$ and $\mathcal{E}_2$ have occurred. we show that \edgeestimation{} outputs $(1 \pm \approxerror)$-approximation of $\edgecount(G)$ when both $\event_1$ and $\event_2$ holds. we have,
    \begin{equation*}\label{eq: 1}
        \sum_{i = 1}^t \approxedgecount(G_i) \in \fbrac{1 \pm \frac{\approxerror}{3}} \sum_{i = 1}^t \edgecount(G_i) \text{  and  } \frac{\onto{\qsetno }{\hyperlen}}{t} \cdot \sum_{i = 1}^t \edgecount(\graph_i) \in \fbrac{1 \pm  \frac{\approxerror}{3}} \edgecount 
    \end{equation*}
    Combining the above, we get
    \begin{align*}
        \fbrac{1 - \frac{\approxerror}{3}}^2 \edgecount(G) \leq 
        \frac{\onto{\qsetno }{\hyperlen}}{t} \sum_{i = 1}^t \approxedgecount(G_i)
        \leq \fbrac{1 + \frac{\approxerror}{3}}^2
        \edgecount(G)
    \end{align*}
    Using the fact that $1 - \approxerror \leq (1 - \approxerror)^2$, $ (1 + \approxerror)^2 \leq 1 + \approxerror$, for any $\approxerror \in (0, 1/3)$, we get 
    \begin{align*}
        \fbrac{1 - \approxerror} \edgecount(G) \leq 
        \frac{\onto{\qsetno }{\hyperlen}}{t} \sum_{i = 1}^t \approxedgecount(G_i)
        \leq \fbrac{1 + \approxerror}
        \edgecount(G)
    \end{align*}

    \noindent\textbf{Query Complexity:} \edgeestimationISC{} is executed $t = O(f^2(\qsetno, \hyperlen)/\approxerror^2)$ times. So, from Lemma~\ref{lem:final:query} the total query expected complexity is 
    \begin{equation*}
        O_{\hyperlen}\fbrac{ f^2(\qsetno, \hyperlen) \log^2 \vertexcount \cdot \log \log n \cdot \log^2(1/\approxerror)/\approxerror^4} = C_{\qsetno, \hyperlen} \cdot \log^2 \vertexcount \cdot \log \log n \cdot \log^2(1/\approxerror)/\approxerror^4
    \end{equation*}
    where $C_{\qsetno, \hyperlen}$ is function in terms of $\qsetno$ and $\hyperlen$.
\end{proof}




\section{Lower Bounds}\label{Section: Lower Bounds}

In this section, we establish the main lower bound of our work. We show that for every ($1 \le s < d$), any randomized algorithm with access only to the $s$-colorful decision oracle (\cisq{}) must make a polynomial number of queries to estimate the number of hyperedges in a $d$-uniform hypergraph. 

We prove the following theorem in this section which is stronger that \Cref{thm:lower-intro} in the sense that this is parametrized by the number of edges.
\begin{theorem}[Main lower bound]
\label{thm:main-lower}
Let $d\ge 2$, $1\le s<d$, and let $m,n\in\mathbb{N}$ satisfy
$1\le m\le \binom{n}{d}$.
Any randomized algorithm that has access to an $s$-colorful $\cisq$ oracle must make
\[
\widetilde{\Omega}_{d,s}\!\left(
\min\left\{
m^{\frac{1}{d-1}\cdot\frac{d-s}{d-s+1}},
\left(\frac{n^d}{m}\right)^{\frac{d-s}{d-s+1}}
\right\}
\right)
\]
queries to distinguish whether a rank-$d$ hypergraph has at most $m$ hyperedges or at least $2m$ hyperedges, with probability at least $2/3$.
\end{theorem}

\Cref{thm:main-lower} immediately implies \Cref{thm:lower-intro}. Indeed, setting $m=\Theta(n^{d-1})$ in \Cref{thm:main-lower} yields a lower bound of
\[
\widetilde{\Omega}_{d,s}\!\left(
\left(\frac{n^d}{m}\right)^{\frac{d-s}{d-s+1}}
\right)
=
\widetilde{\Omega}_{d,s}\!\left(
n^{\frac{d-s}{d-s+1}}
\right),
\]

We remark that the lower bound in \Cref{thm:main-lower} continues to hold even if the algorithm is given access to a stronger $s$-colorful decision oracle ($\cisq{}$), in which the $s$ queried sets are allowed to overlap. For simplicity of notation, throughout this section we continue to use $\cisq{}$ to denote this stronger oracle.

We prove \Cref{thm:main-lower} by proving \Cref{lem:dense} and~\Cref{lem:sparse}, stated below. Note that Lemmas~\ref{lem:dense} and~\ref{lem:sparse} establish the desired lower bound when the number of hyperedges is large ($m\ge c  n^{d-1} / (2d)^{d-1}$) and small ($m<cn^{d-1}/(2d)^{d-1}$), respectively. Here $c$ is a suitable  constant.

\begin{lemma}[Lower bound for dense hypergraphs]
\label{lem:dense}
Let $m\ge c n^{d-1}/d^{d-1}$. Any randomized algorithm that has access to an $s$-colorful $\cisq$ oracle must make
\[
{C}_{s,d}\cdot 
\left(\frac{n^d}{m}\right)^{\frac{d-s}{d-s+1}}
\cdot \frac{1}{\log^6 n}
\]
queries to distinguish whether a rank-$d$ hypergraph has at most $m$ hyperedges or at least $2m$ hyperedges, with probability at least $2/3$. Here $C_{s,d}$ is a constant that depends only on $s$ and $d$.
\end{lemma}

\begin{lemma}[Lower bound for sparse hypergraphs]
\label{lem:sparse}
Let $m<c n^{d-1}/d^{d-1}$. Any randomized algorithm that has access to an $s$-colorful $\cisq$ oracle must make
\[
C'_{s,d} \cdot 
m^{\frac{1}{d-1}\cdot\frac{d-s}{d-s+1}} \cdot \frac{1}{\log^6 n}
\]
queries to distinguish whether a rank-$d$ hypergraph has at most $m$ hyperedges or at least $2m$ hyperedges, with probability at least $2/3$. Here $C_{s,d}$ is a constant that depends only on $s$ and $d$.
\end{lemma}
We first prove Lemma~\ref{lem:dense} in We first prove \Cref{lem:dense} in \Cref{sec:dense}. We then establish \Cref{lem:sparse} in \Cref{sec:sparse}, whose proof relies on \Cref{lem:dense}. Here $\alpha_{s,d}$ is a constant that depends only on $s$ and $d$.

\subsection{Proof of the lower bound in dense hypergraphs}\label{sec:dense}

For technical reasons,  we assume $m\leq \frac{\constant_{u}\vertexcount^d}{d^d\log^6\vertexcount}$ for some sufficiently small constant $\constant_u$. Otherwise, the claimed lower bound is straightforward as $\widetilde{\Omega}_{s,d}(\cdot)$ hides polylog term. 

To establish the lower bound for the case of dense hypergraphs, we argue using Yao's principle \cite{YaoMinMax}. We design two hypergraph distributions that are indistinguishable by any deterministic algorithm that uses less than $\smallo{\numq}$ queries, where:
\begin{align*}
    \numq = \fbrac{\frac{n^{d}}{\edgecount}}^{\frac{d-s}{d-s+1}}\cdot\frac{1}{\log ^3 \vertexcount}
\end{align*}

We start by describing the two distributions, $\dyes$ and $\dno$. Without loss of generality, we assume the vertex set of both the graphs to be $[\vertexcount]$ for some $\vertexcount \in \Nat$. 

\begin{itemize}
    \item{$\dyes$:} The vertex set is partitioned into $\sequence{A}{d}$ by putting each vertex in one of the $\sbrac{A_i}_{i \in [d]}$ independently and uniformly at random. For every $d$-tuple $\sequence{\vertex}{d}$ such that $\vertex_i \in A_i$ for all $i \in [d]$, the graph has an edge $\ang{\sequence{\vertex}{d}}$ with probability:
    $$\frac{d^d\edgecount}{4\vertexcount^d}$$ 

    \item{$\dno$:} The graph is first constructed as is done for the case of $\dyes$. Additionally, we construct a set of vertices $B$ by choosing each vertex of $A_1$ independently with probability:
    $$\frac{(2d)^d\edgecount\log\vertexcount}{\vertexcount^{d}}$$
    For every $\vertex_b \in B$, we add edges $\big\{\ang{\vertex_b,a_2,a_3,\ldots,a_d}|a_2 \in A_2, a_3 \in A_3, \ldots a_d \in A_d\big\}$.
\end{itemize}
Both the probabilities above are less than $1$ as we are working with $m\leq \frac{c_un^d}{d^d \log^6 n}$.

We show (in \Cref{lem:edge}) that, with probability $1-o(1)$, graphs sampled from $\dno$ contain almost twice as many edges as graphs sampled from $\dyes$. However, any algorithm that distinguishes between $\dyes$ and $\dno$ must encounter at least one vertex in the set $B$. We prove that, unless the algorithm makes $\Omega(q)$ queries, it encounters a vertex in $B$ with probability only $o(1)$. This forms the core of our lower-bound argument, which we formalize in the remainder of this section.



\begin{lemma}\label{lem:edge}
    With probability $1 - \smallo{1}$, 
    \begin{enumerate}
        \item A graph sampled from $\dyes$ has at most ${\frac{\edgecount}{2}}$ edges, and
        \item A graph sampled from $\dno$ has at least $\edgecount$ edges.
    \end{enumerate}
\end{lemma}

\begin{proof}
    We proceed by analysing $\dyes$ and $\dno$ separately.
    \paragraph{$\dyes$:} For graphs generated by $\dyes$, there are at most $\frac{\vertexcount^d}{d^d}$ $d$-tuples $\ang{\sequence{\vertex}{d}}$ such that $\vertex_i \in A_i$ for all $i \in [d]$. Each of which forms an hyperedge with probability $\frac{d^d\edgecount}{4\vertexcount^d}$. Hence, the expected number of edges is $\frac{\edgecount}{4}$. As each hyperedge is formed independently, the claim follows by a Chernoff bound argument.

    \paragraph{$\dno$:} For graphs generated by $\dno$, we first note that with probability $1 - \smallo{1}$, we have $\size{A_i} \geq \frac{\vertexcount}{2d}$ for all $i \in [d]$. Furthermore, with probability $1 - \smallo{1}$, we have $\size{B} \geq \frac{(2d)^{d-1}\edgecount}{\vertexcount^{d-1}}$. This is true since we assume $m \geq \frac{\constant\vertexcount^{d-1}}{(2d)^{d-1}}$. Note that each vertex in $B$ forms an edge with all $(d-1)$-tuples of the form $\ang{\vertex_2,\vertex_3,\ldots,\vertex_d}$ such that $\vertex_i \in A_i$ for all $i$. Hence, by a union bound, we have that the graph has at least ${\edgecount}$ edges with probability $1- \smallo{1}$.
\end{proof}

Next, we want to show that the probability that an algorithm using $\smallo{\numq}$ queries hits a vertex in $B$ is very small. However, such a claim is obviously false if the algorithm makes the queries set of vertices arbitrarily large. However, if the query set is too large, the hypergraph construction ensures that it is not independent with high probability. Hence, we can restrict \emph{meaningful} queries to be of bounded size, say $\lbthres$ and argue based on that. For this purpose, we  describe an augmented oracle that simulates the \cisq{} query for graphs generated by $\dyes$ or $\dno$.

 At each stage, an algorithm with access to the augmented oracle can maintain a triplet $\oracle$ where $\ocvset \subseteq [\vertexcount]$ is a set of vertices, $\func{\oclbl}{\ocvset}{\sequence{A}{d},B}$ is a labelling function that denotes the set containing the vertices, and $\func{\oceset}{\binom{\ocvset}{d}}{\sbrac{0,1}}$ is the indicator function denoting whether the corresponding hyperedge exists or not. Initially, $\oracle$ is initiated with $\ocvset = \emptyset$. The augmented oracle takes as input the current state of $\oracle$ and a set of vertices $\sequence{Q}{s} \subseteq [\vertexcount]$, and returns an updated $\oracle$ with additional information updated as outlined below depending on the size of query sets $\size{Q_i}$, and a threshold $\lbthres$ where:
\begin{align*}
    \lbthres = \frac{1}{s}\fbrac{\frac{\vertexcount^d}{\edgecount}}^{\frac{1}{d-s+1}}\log \vertexcount
\end{align*}

\begin{enumerate}
    \item \textbf{If $\exists \sT \subseteq [s]$ such that no colorful $\size{\sT}$-tuple can be formed using $\sbrac{Q_t}_{t\in\sT}$:} First, we consider the case that a subset of queries $\sbrac{Q_t}_{t\in\sT}$ prohibits the formation of a colorful tuple, i.e. there is no $\sequence{\vertex}{\size{\sT}}$ such that each $\vertex_i$ belongs to a different color class in $\sequence{A}{d}$. This implies that the query set is an independent set according to the definition of the \cisq{} oracle. The augmented oracle chooses a minimal such tuple $\sT^*$, and returns an updated $\oracle$ with $\ocvset \gets \ocvset\cup \sbrac{\cup_{i \in \sT^*} Q_i}$, and updates $\oclbl$ for vertices in $\cup_{t \in \sT^*} Q_t$.  
    
    \item \textbf{$\size{Q_i} \leq \lbthres$ for all $i \in [s]$:} The augmented oracle returns $\ocvset \gets \ocvset \cup \sbrac{\cup_{i \in [s]} Q_i}$, and $\oclbl$ and $\oceset$ are updated for the vertices in  $\ocvset \cup \sbrac{\cup_{i \in [s]} Q_i}$ according to $\sequence{A}{d}$, $B$, and the set of hyperedges in the sampled graph $\graph$.
    
    \item \textbf{$\exists i\in [s]$ such that $\size{Q_i} > \lbthres$:} The augmented oracle samples a subset $R_i \subseteq Q_i$ of size $\lbthres$ uniformly at random, for all $i \in [s]$ such that $\size{Q_i} > \lbthres$. For all $i \in [s]$ such that $\size{Q_i} \leq \lbthres$, denote $R_i = Q_i$. If $\sequence{R}{s}$ is not an independent set, the oracle returns $\ocvset \gets \ocvset\cup \sbrac{\cup_{i \in [s]} R_i}$, and returns $\oclbl$ and $\oceset$ on the updated set of vertices. If $\sequence{R}{s}$ is an independent set, the oracle \emph{fails}, and returns the entire graph, i.e. $\ocvset \gets [\vertexcount]$, and $\oclbl$ and $\oceset$ updated the vertices in  $\ocvset \cup \sbrac{\cup_{i \in [s]} R_i}$ according to $\sequence{A}{d}$, $B$, and the set of hyperedges in the sampled graph $\graph$.
\end{enumerate}

Note that, without loss of generality, we can assume that all the queries consists of vertices not explored already, i.e. $Q_i \subseteq [n]\setminus \ocvset$.
Next, we show that the augmented oracle can simulate the original \cisq{} queries. We formalize this in the following result.

\begin{lemma}
    Given a deterministic algorithm $\alg$ with access to \cisq{} queries that makes $\numq$ queries, there exists an algorithm $\augalg$ with access to augmented oracle such that:
    \begin{itemize}
        \item $\augalg$ makes $\numq$ queries to the augmented oracle.
        \item $\augalg(\graph) = \alg(\graph)$
    \end{itemize}
\end{lemma}

\begin{proof}
    We show how to simulate the \cisq{} queries using augmented queries. The algorithm $\augalg$ then executes the algorithm $\alg$ using the simulated $\cisq{}$ queries. We consider the cases whether $\size{Q}$ is less than or equal to $\lbthres$, or greater than $\lbthres$ separately:
    \begin{itemize}
        \item \textbf{If $\exists \sT \subseteq [s]$ such that no colorful $\size{\sT}$-tuple can be formed using $\sbrac{Q_t}_{t\in\sT}$:} The updated triplet contains the labels for vertices in $\cup_{t \in \sT^*} Q_t$. Given the updated labels, it can be checked that  the query sets $ \sbrac{\cup_{t \in \sT^*} Q_t}$ does not form a colorful tuple, i.e. the query forms an independent set.
    
        \item \textbf{$\size{Q_i} \leq \lbthres$ for all $i \in [s]$:} The updated triplet $\oracle$ returned by the augmented oracle can be directly used to determine whether the queries vertices form an independent set or not based on the edges $\oceset$.  
        
        \item \textbf{$\exists i\in [s]$ such that $\size{Q_i} > \lbthres$:} We consider the two possible cases for the augmented oracle depending on whether $\sequence{R}{s}$ is an independent set or not.
        \begin{itemize}
            \item If $\sequence{R}{s}$ is not an independent set, then $\sequence{Q}{s}$ is also not an independent set.
            \item If $\sequence{R}{s}$ is an independent set, then the augmented oracle reveals the graph completely. Hence, one can decide whether $\sequence{Q}{s}$ is an independent set or not.
        \end{itemize}
    \end{itemize}    
\end{proof}


We now focus on analyzing $\augalg$, and bound its ability to distinguish between graphs generated by $\dyes$ and $\dno$ using $\numq$ queries. Observe that $\augalg$ can be seen as a tree  of depth $\numq$ whose each edge is labeled by a knowledge triplet $\oracle$, and each node is labeled by a query set $\sequence{Q}{s}$. Given an edge $(\altvertex,\vertex)$, we define the knowledge triplet of the node $\vertex$ to be the knowledge triplet associated with the incoming edge $(\altvertex,\vertex)$. Furthermore, each leaf node is additionally labeled either \emph{`yes'} or \emph{`no'} denoting the output of the algorithm $\augalg$. Below, we define the notion of \emph{good} nodes that would be useful in our proof.

\begin{definition}[Good and Bad Nodes]
    We say a node $\vertex$ in the tree is good if its corresponding knowledge triplet $\oracle$ satisfies $\oclbl^{-1}(B) = \emptyset$, and say the node $\vertex$ to be  bad otherwise.
\end{definition}

Now, we bound the probability that $\augalg$ reaches a bad node. Note that $\augalg$ can not reach a bad node if the graph $\graph$ is generated by $\dyes$. Hence, we only bound the probability for the case when the graph $\graph$ is generated by $\dno$. Next, we outline the conditional distribution of the graph $\graph$ given the algorithm $\augalg$ is at a node $\vertex$. The fact shows that given the algorithm is at a good node $\vertex$, the conditional distribution of the unexplored graph retains sufficient randomness.

\begin{fact}
    Let the algorithm $\augalg$ arrive at a good node $\vertex$ with corresponding knowledge triplet $\oracle$. Then, the conditional distribution of the graph $\graph$ is as follows: 

    \begin{itemize}
        \item \textbf{$\graph \sim \dyes$:} Assign each $i \in \ocvset$ to $\sequence{A}{d}$ as fixed in $\oclbl$, and add edges as fixed in $\oceset$. Assign vertices in $[\vertexcount]\setminus\ocvset$ uniformly at random to one of $A_1,\ldots,A_d$. For any pair $(\sequence{a}{d})$ such that $\exists i: a_i \notin \ocvset$, add an edge $(\sequence{a}{d})$ with probability $\frac{d^d\edgecount}{4\vertexcount^d}$.

        \item \textbf{$\graph \sim \dno$:} Assign each $i \in \ocvset$ to $\sequence{A}{d}$ as fixed in $\oclbl$, and add edges as fixed in $\oceset$. Assign vertices in $[\vertexcount]\setminus\ocvset$ uniformly at random to one of $A_1,\ldots,A_d$. For each $i \in A_1$ such that $i \notin \ocvset$, include $i$ in $B$ independently with probability $\frac{(2d)^d\edgecount\log\vertexcount}{\vertexcount^d}$. For any pair $(\sequence{a}{d})$ such that $\exists i$ such that $a_i \notin \ocvset$, add an edge $(\sequence{a}{d})$ with probability $\frac{d^d\edgecount}{4\vertexcount^d}$. For every $\vertex_{B\setminus\ocvset} \in B\setminus \ocvset$, we add edges $\big\{\ang{\vertex_{B\setminus\ocvset},a_2,a_3,\ldots,a_d}|a_2 \in A_2, a_3 \in A_3, \ldots a_d \in A_d\big\}$.
    \end{itemize}
\end{fact}

Next, we show that if the query sets does not form a colorful tuple, then all of the sets have their size bounded by $\lbthres$ with high probability. Note that, here the randomness comes from generating the graphs according to the underlying distribution.

\begin{lemma}\label{Lem: No Colorful Implies Small}
    Let $Q_i$ be a set such that $\size{Q_i} > \lbthres$, and $\sT^* \subseteq [s]$ be a minimal set of query sets $\sbrac{Q_t}_{t\in\sT^*}$ such that no colorful $\size{\sT^*}$-tuple can be formed using them. Then, with probability $1-\smallo{\frac{1}{\numq}}$, $i \notin \sT^*$.
\end{lemma}

\begin{proof}
    We prove by contradiction. Assume that $i \in \sT^*$ such that $\size{Q_{i}} > \lbthres$. We first show that any such query set $Q_i$ intersects with all the colored partitions $\sequence{A}{d}$ with high probability.
    
    Recall that the vertices are assigned to the partitions $\sequence{A}{d}$ uniformly, and independently at random. The probability that $Q_{i}$ fails to intersect a specific partition $A_j$ is bounded by $\left(1 - \frac{1}{d}\right)^{\size{Q_{i}}}$. By a union bound, the probability that $Q_{i}$ misses any of the $d$ partitions is at most:
    \begin{align*}
        d \left(1 - \frac{1}{d}\right)^{\size{Q_{i}}} \leq d e^{-\frac{\size{Q_{i}}}{d}} < d e^{-\frac{\lbthres}{d}} \leq d \exp\left(-\frac{1}{sd} \numq^{\frac{1}{d-s}}\right) = \smallo{\frac{1}{q}}
    \end{align*}
    
    Where the second last inequality follows from the fact that $\lbthres = \frac{1}{s}\fbrac{\frac{\vertexcount^d}{\edgecount}}^{\frac{1}{d-s+1}}$ and $\numq = \fbrac{\frac{\vertexcount^d}{\edgecount}}^{\frac{d-s}{d-s+1}}$, and the last equality follows from the fact that $d > s$.

    Now, we show that if a query set spans all colors, it can not be part of a minimal set of queries that can not form a colorful tuple. Consider the bipartite graph $(\bT,\bA,\edgeset)$ where $\bT = \sbrac{Q_t}_{t \in \sT^*}$, $\bA = \sbrac{\sequence{A}{d}}$, and $(Q_t,A_i) \in \edgeset$ if $Q_t \cap A_i \neq \emptyset$. If there exists a matching of $\size{\sT^*}$ in this graph, then there exists a colorful $\size{\sT^*}$-tuple that can be formed using $\sbrac{Q_t}_{t \in \sT^*}$. Hence, we consider the case where no matching of $\size{\sT^*}$ exists in this graph. By Hall's Marriage Theorem, and the minimality of $\sT^*$, we have:
    \begin{align*}
        \size{\neighbour{(\sT^*)}} < \size{\sT^*}
    \end{align*}
    Note that the query set $Q_{i}$ is part of $\sT^*$, and hence with probability $1- \smallo{\frac{1}{\numq}}$,
    \begin{align*}
        \size{\neighbour{(\sT^*)}} = d
    \end{align*}
    On the other hand, we have:
    \begin{align*}
        \size{\sT^*} \leq s < d
    \end{align*}
    Hence, a contradiction.  
\end{proof}

Now, we bound the probability that $\augalg$ reaches a bad node. Note that $\augalg$ can not reach a bad node if the graph $\graph$ is generated by $\dyes$. Hence, we only bound the probability for the case when the graph $\graph$ is generated by $\dno$.

\begin{lemma}
        $\Pr\limits_{\graph\sim\dno}\tbrac{\augalg(\graph) \text{ reaches a bad node}} = \smallo{1}$
\end{lemma}

\begin{proof}
    Note that the root note is always good as $\ocvset = \emptyset$. We now bound the probability that a query made by the algorithm $\augalg$ causes it to reach a bad node given it started from a good node $\vertex$. We consider the cases whether $\size{Q_i} \leq \lbthres$ for all $i \in [s]$ or $\size{Q_i} > \lbthres$ for some $i \in [s]$ separately.

    \paragraph{$\size{Q_i} \leq \lbthres$ for all $i \in [s]$:} In this case, the query sets $Q_i$ must share a vertex with $B$. Hence, we have:
    \begin{align*}
        &\Pr[Q_i \cap B \neq \emptyset|\augalg \text{ is at a \emph{good} node }\vertex]\\
        \leq &\size{Q_i}\cdot\frac{1}{d}\cdot\frac{(2d)^d\edgecount\log\vertexcount}{\vertexcount^d}\\
        \leq &\frac{1}{s}\fbrac{\frac{\vertexcount^d}{\edgecount}}^{\frac{1}{d-s+1}}\log\vertexcount\cdot\frac{1}{d}\cdot\frac{(2d)^d\edgecount\log\vertexcount}{\vertexcount^d}\\
        = &\frac{1}{s}\fbrac{\frac{\vertexcount^d}{\edgecount}}^{\frac{s-d}{d-s+1}}\cdot d^{d-1}\log^2\vertexcount\\
        = \frac{1}{s}&\smallo{\frac{1}{\numq}}
    \end{align*}
    Hence, by union bound, no query set $Q_i$ intersect with $B$ with probability $1 - \smallo{\frac{1}{\numq}}$.

    \paragraph{If $\exists \sT \subseteq [s]$ such that no colorful $\size{\sT}$-tuple can be formed using $\sbrac{Q_t}_{t\in\sT}$:} By \Cref{Lem: No Colorful Implies Small}, we have that any such query can contain a $Q_i$ such that $\size{Q_i} >  \lbthres$ with probability at most $\smallo{\frac{1}{\numq}}$. Hence, all such queries reveals $\cup_{i \in \sT^*} Q_i$ such that $\size{Q_i} < \lbthres$ for all $i \in \sT^*$ with probability at least $1 - \smallo{\frac{1}{\numq}}$. Conditioning on the fact, we can bound the probability of reaching a bad node by $\smallo{\frac{1}{\numq}}$ using the same argument as the previous case. 

    \paragraph{$\exists i\in [s]$ such that $\size{Q_i} > \lbthres$:} In this case, there are two ways for $\augalg$ to reach a bad node: 1) If $R_i \cap B \neq \emptyset$ for some $i \in [s]$, and 2) If $\sequence{R}{s}$ is an independent set. The probability of the first event can be bounded as above. Hence, we focus on the second event. Observe that in this case there exits an $i$ such that $\size{R_i} \geq  \lbthres$. Let us denote $R_{i^*}$ to be the largest query set $R_i$ satisfying $\size{R_{i^*}} \geq \lbthres$. Observe that with probability $1 - \frac{1}{\poly{\vertexcount}}$, we have that $A_j \cap R_{i^*}= \bigomega{\frac{\lbthres}{d}}$ for all $j \in [d]$. Now, we consider the number of edges $(\sequence{\vertex}{d})$ such that $\fbrac{\sequence{\vertex}{{s-1}}}$ forms a colorful tuple formed by $\sbrac{R_i}_{i \in d} \setminus R_{i^*}$, and $\vertex_{s},\vertex_{s},\ldots,\vertex_d \in R_{i^*}$. The existence of valid $(s-1)$-colorful tuple ensures that the expected number of such edges is at least:
    \begin{align*}
        \prod_{i = s}^d \bigomega{\frac{\lbthres}{d}} = \bigomega{\fbrac{\frac{\lbthres}{d}}^{d-s+1}} = \bigomega{\numq}
    \end{align*}

    By Chernoff Bound, we have that with probability $1 - \smallo{\frac{1}{\numq}}$, there exists such an edge. Hence, by an union bound over all possible queries $\numq$, the algorithm $\augalg$ reaches a bad node with probability at most $\smallo{1}$.
\end{proof}

Next, we show that the probability of hitting a good node is not higher for a graph sampled from $\dno$ compared to that sampled from $\dyes$.

\begin{lemma}
    For every good node $\vertex$ in the tree of $\augalg$, we have:
    \begin{align*}
        \Pr_{\graph \sim \dno} [\augalg(\graph) \text{ reaches } \vertex] \leq \Pr_{\graph \sim \dyes} [\augalg(\graph) \text{ reaches } \vertex]
    \end{align*}
\end{lemma}

\begin{proof}
    We first show a coupling between the distributions $\dyes$, and $\dno$. First, draw a graph $\graph$ from $\dyes$ and obtain corresponding labeling $\sequence{A}{d}$. Then, we draw $r \in \sbrac{0,1}^{\size{A_1}}$ where each bit is set to $1$ independently with probability $\frac{d^d\edgecount\log\vertexcount}{\vertexcount^{d}}$. Next, we construct $\graph'$ from $\graph$ by adding edges $\big\{\ang{i,a_2,a_3,\ldots,a_d}|a_2 \in A_2, a_3 \in A_3, \ldots a_d \in A_d\big\}$ for all $i \in A_1$ such that $r_i = 1$. This defines a coupling between $\dyes$ and $\dno$ as $\graph'$ has the marginal distribution $\dno$.

    Now, let us consider $\vertex$ to be a good node in the execution tree of the algorithm $\augalg$, and let $\oracle$ be its corresponding knowledge triple. $\augalg(\graph')$ reaches $\vertex$ if $\augalg(\graph)$ reaches $\vertex$ and $r_i = 0$ for all $i \in \oclbl^{-1}(A_1)$. Hence, we have:
    \begin{align*}
        \Pr_{(\graph,\graph')} [\augalg(\graph') \text{ reaches } \vertex]
        &= \Pr_{(\graph,\graph')} [\augalg(\graph) \text{ reaches } \vertex \text{ and } r_i = 0, \forall i \in \oclbl^{-1}(A_1)]\\
        &\leq \Pr_{(\graph,\graph')} [\augalg(\graph) \text{ reaches } \vertex]
    \end{align*}
\end{proof}

Now we show that an algorithm $\augalg$ using augmented oracle can not distinguish between graphs generated from $\dyes$ and $\dno$.

\begin{lemma}
    Let $\augalg$ be a deterministic algorithm using at most $\numq$ augmented oracles. Then, we have:
    \begin{align*}
        \Pr_{\graph\sim\dno}[\augalg(\graph) \text{ outputs ``no"}] -
        \Pr_{\graph\sim\dyes}[\augalg(\graph) \text{ outputs ``no"}] \leq \smallo{1}
    \end{align*}
\end{lemma}

\begin{proof}
    Let $\sL$ denote the set of leaves in the execution tree of $\augalg$ that are labeled ``no". Let $\sL_{g} \subseteq \sL$ denote the subset of good nodes in $\sL$. Then, we have:
    \begin{align*}
        \Pr_{\graph\sim\dno}[\augalg(\graph) \text{ outputs ``no"}]
        &\leq \sum_{\vertex \in \sL_{g}} \Pr_{\graph\sim\dno}[\augalg \text{ reaches $\vertex$}] + \Pr_{\graph\sim\dno}[\augalg(\graph) \text{ reaches a bad node}]\\
        &\leq \sum_{\vertex \in \sL_{g}} \Pr_{\graph\sim\dyes}[\augalg \text{ reaches $\vertex$}] + \smallo{1}\\
        &\leq \Pr_{\graph\sim\dyes}[\augalg(\graph) \text{ outputs ``no"}] + \smallo{1} 
    \end{align*}

\end{proof}

{Note that the proof of Lemma 20 is done.}

                  
\subsection{Proof of the lower bound in sparse hypergraphs}\label{sec:sparse}
Let $\alpha = \frac{d-s}{d-s+1}$.
Suppose, for the sake of contradiction, that there exists an algorithm $\mathcal A$ which, for every $d$-uniform hypergraph $H_1$ on $\Theta(n)$ vertices satisfying $m< c\,n^{d-1}/(2d)^{d-1}$ (the sparse case), distinguishes between $|E(H_1)|\le m$ and $|E(H_1)|\ge 2m$ using $s$-colorful $\cisq$ queries with probability at least $2/3$, while making fewer than
\[
C'_{s,d}\cdot \frac{m^{\alpha/(d-1)}}{\log^6n}
\]
queries, where $C'_{s,d}$ is a constant to be fixed later.
We use $\mathcal A$ to contradict \Cref{lem:dense}. Define $n_2:=\big((2d)^{d-1}m/c\big)^{1/(d-1)}$, so that $m=c\,n_2^{d-1}/(2d)^{d-1}$. That is, $(n_2,m)$ sits exactly on the dense/sparse threshold for $n_2$, and in particular satisfies (with equality) the hypothesis $m\ge c\,n_2^{d-1}/(2d)^{d-1}$ of \Cref{lem:dense}. Moreover, since $m< c\,n^{d-1}/(2d)^{d-1}$, we get $n_2^{d-1}=(2d)^{d-1}m/c< n^{d-1}$, so $n_2<n$. Therefore the pair $(n_2,m)$ satisfies the assumptions of \Cref{lem:dense}.
Let $H_2$ be an arbitrary instance of \Cref{lem:dense}, and construct $H_1:=H_2\sqcup H_0$, where $H_0$ is a vertex-disjoint $d$-uniform hypergraph containing $\Theta(n)$ additional vertices and no hyperedges. The added vertices respect the given $s$-color partition. Thus, $H_1$ has $\Theta(n)$ vertices and exactly the same number of hyperedges as $H_2$.

Every $s$-colorful $\cisq$ query made by $\mathcal A$ to $H_1$ can be simulated using one query to $H_2$. Suppose $\mathcal A$ queries the sets $(S_1,\ldots,S_s)$. For every $i\in[s]$, let $S_i':=S_i\cap V(H_2)$. Since the sets $(S_1,\ldots,S_s)$ are pairwise disjoint, so are $(S_1',\ldots,S_s')$, and therefore they form a valid $s$-colorful $\cisq$ query to $H_2$. Furthermore, $H_0$ contains no hyperedges, and there are no hyperedges between $H_0$ and $H_2$. Hence, the answer to the query $(S_1,\ldots,S_s)$ on $H_1$ is identical to the answer to $(S_1',\ldots,S_s')$ on $H_2$. Consequently, $\mathcal A'$ simulates every query of $\mathcal A$ using exactly one oracle query, and therefore has the same success probability and query complexity as $\mathcal A$.

A direct computation gives $n_2^d/m = \big((2d)^{d-1}/c\big)^{d/(d-1)} m^{1/(d-1)}$, so that, writing $C=\big((2d)^{d-1}/c\big)^{d/(d-1)}$ for a constant depending only on $d$,
\[
\left(\frac{n_2^d}{m}\right)^{\alpha}=C^{\alpha}\,m^{\alpha/(d-1)}.
\]
Hence, the query bound assumed for $\mathcal A$ (and hence for its simulation $\mathcal A'$ on $H_2$) can be rewritten as
\[
C'_{s,d}\cdot \frac{m^{\alpha/(d-1)}}{\log^{c}n}
=
\frac{C'_{s,d}}{C^{\alpha}}\left(\frac{n_2^d}{m}\right)^{\alpha}\frac{1}{\log^{6}n}.
\]
We now set $C'_{s,d}:=C^{\alpha}\cdot C_{s,d}$, where $C_{s,d}$ is the constant appearing in \Cref{lem:dense}.
On the other hand, \Cref{lem:dense} asserts that every algorithm for distinguishing whether $|E(H_2)|\le m$ or $|E(H_2)|\ge 2m$ requires at least $C_{s,d}\big(n_2^d/m\big)^{\alpha}/(\log n)^{6}$ queries. This directly contradicts the query bound assumed for $\mathcal A'$. Therefore, $\mathcal A$ cannot exist, proving the desired lower bound.
Hence, we are done with the proof of \Cref{lem:sparse}.
\section{Discussion and Applications}\label{section: application}

Our main algorithm naturally yields a high-probability constant-factor approximation to the number of hyperedges. Such a coarse estimator is a fundamental primitive in the framework of \cite{DellLM22} for hyperedge estimation and almost-uniform hyperedge sampling. We observe that this framework, although originally developed for the $d$-colorful decision oracle, applies verbatim to the $s$-colorful random certificate oracle. Consequently, our main result immediately implies efficient algorithms for both problems.

By standard probability amplification, Theorem~\ref{thm:main-intro} immediately yields the following.

\begin{corollary}\label{cor:constant-factor}
For every integer $1\le s\le d$, there exists a randomized algorithm that, given access to the $s$-colorful random certificate oracle $\ciscq$, outputs an estimate $\widehat{m}$ satisfying $\frac{m}{2}\le \widehat{m}\le 2m$ with probability at least $1-1/\mathrm{poly}(n)$ using $O_{s,d}(\log^3 n\log\log n)$ $s$-colorful random certificate queries.
\end{corollary}

Besides hyperedge estimation, another fundamental task is \emph{almost-uniform hyperedge sampling}, where the goal is to output a hyperedge according to a distribution $\mathcal{D}$ satisfying
\[
\Pr_{e\sim\mathcal{D}}[e]\in
\left[
\frac{1}{(1+\varepsilon)m},
\frac{1+\varepsilon}{m}
\right]
\]
for every hyperedge $e$.

As observed by \cite{BhattacharyaBGM24}, the framework of \cite{DellLM22} reduces both hyperedge estimation and almost-uniform hyperedge sampling to the problem of obtaining a high-probability constant-factor estimate of the number of hyperedges.

\begin{proposition}[{\cite{DellLM22,BhattacharyaBGM24}}]
\label{prop:dell-framework}
Suppose there exists an algorithm {\sc Coarse} that, given access to the $s$-colorful random certificate oracle, makes $q$ queries and outputs an estimate $\widehat{m}\in[m/b,bm]$ with probability at least $1-1/\mathrm{poly}(n)$. Then:
\begin{enumerate}
    \item Hyperedge estimation can be solved using $\widetilde{O}_d\!\left((\log(nb)+b^2\varepsilon^{-2}\log^2 n)\log^2 n\cdot q\right)$ queries.
    \item Almost-uniform hyperedge sampling can be solved using $\widetilde{O}_d\!\left((\log(nb)+b^2\varepsilon^{-2}\log^2 n)\log^6 n\cdot q\right)$ queries.
\end{enumerate}
\end{proposition}

Instantiating Proposition~\ref{prop:dell-framework} with the coarse estimator from Corollary~\ref{cor:constant-factor} immediately yields the following theorem.

\begin{theorem}\label{thm:hyperedge-sampling}
For every integer $1\le s\le d$, there exists a randomized algorithm that, given access to the $s$-colorful random certificate oracle $\ciscq$, computes a $(1\pm\varepsilon)$-approximation of the number of hyperedges using $\widetilde{O}_{s,d}(\varepsilon^{-2}\log^7 n\log\log n)$ queries.

Moreover, there exists a randomized algorithm that outputs a hyperedge according to a distribution $\mathcal{D}$ satisfying $\Pr_{e\sim\mathcal{D}}[e]\in\left[\frac{1}{(1+\varepsilon)m},\frac{1+\varepsilon}{m}\right]$ for every hyperedge $e$, using $\widetilde{O}_{s,d}(\varepsilon^{-2}\log^{11} n\log\log n)$ queries.
\end{theorem}

Combining Theorem~\ref{thm:main-intro} with Theorem~\ref{thm:hyperedge-sampling} (or the theorem obtained from Proposition~\ref{prop:dell-framework}) immediately yields the following corollary.

\begin{corollary}
For every integer $1\le s\le d$, there exists a randomized algorithm that, given access to the $s$-colorful random certificate oracle $\ciscq$, computes a $(1\pm\varepsilon)$-approximation of the number of hyperedges with probability at least $1-1/\mathrm{poly}(n)$ using
\[
\widetilde{O}_{s,d}\!\left(
\varepsilon^{-2}\log^3 n\log\log n\cdot
\min\{\varepsilon^{-2},\log^4 n\}
\right)
\]
$s$-colorful random certificate queries.
\end{corollary}

The above corollary combines the strengths of the two approaches. For $\varepsilon\ge 1/\log^2 n$, the direct algorithm of Theorem~\ref{thm:main-intro} yields the better query complexity, whereas for $\varepsilon<1/\log^2 n$, the reduction through the coarse estimator and the framework of \cite{DellLM22} gives the better bound. In particular, while our direct algorithm provides the strongest known bounds for hyperedge estimation over a broad range of approximation parameters, the reduction via coarse estimation additionally yields an efficient algorithm for almost-uniform hyperedge sampling.

\section{Conclusion}
\label{sec:conclusion_and_open_problems}

We introduced a hierarchy of $s$-colorful query oracles for hyperedge
estimation and gave the first systematic study of their power. Our
results show that, for every $1\le s\le d$, a polylogarithmic number of
$s$-colorful random certificate queries suffice for hyperedge
estimation, with the exponent of the $\log n$ factor, is an absolute
constant, independent of both $s$ and $d$. In contrast, for every
$1\le s<d$, any algorithm using only the corresponding $s$-colorful
decision oracle requires polynomially many queries. Combined with the
known bounds for the $d$-colorful decision
oracle~\cite{DellLM22,BhattacharyaBGM24,DellLM24}, this reveals a phase
transition at $s=d$: the decision query complexity drops from
polynomial to polylogarithmic precisely at $s=d$, yet even there a
separation persists, as random certificates provably reduce the
exponent of the $\log n$ factor from $\Theta(d)$ to an absolute
constant. More broadly, our work unifies several previously studied
global query models within a common framework, and exhibits a
fine-grained decision-versus-search separation for approximate
counting: search with uniformly random certificates is exponentially
more powerful than decision access whenever the underlying solver
supports fewer than $d$ restriction classes. Determining the exact
query complexity of the $s$-colorful decision oracle as a function of
$s$, and characterizing the query models in which random certificates
provably help, are natural directions for future work.


\bibliographystyle{alpha}
\bibliography{refs} 


\appendix 

\section*{Appendix}

\section{Threshold Checking}\label{appendix: tc}
In this section, we present the algorithms \iscollision{}, \birthdayestimate{} and the proofs of Lemma~\ref{lemma: collision detection} and Lemma~\ref{lemma: birthdayestimate}.

\begin{algorithm}
    \caption{\iscollision($G, \threshold, \confidence$)}\label{alg: collision}
    \begin{algorithmic}[1]
        \Require \ciscq{} access to the $\qsetno$-partite $\hyperlen$-hypergraph $G$ and a threshold $\threshold$.
        \Ensure $0$ or $1$
        \For{$i \gets 1$ to $O(\log(\confidence^{-1}))$} \label{alg: collision line 1}
            \State  Uniformly and independently sample $\preresamplesize = \sqrt{2\threshold}$ random hyperedges from $G$ using \ciscq{} query.\label{alg: collision line 2}
        \If{there is a collision}
            \State $\collision_i \gets 1$
        \Else 
            \State $\collision_i \gets 0$
        \EndIf
        \EndFor
        \If{more than half of $\collision_i$ are $1$}
            \State \Return  $1$
        \Else 
            \State \Return $0$
        \EndIf
    \end{algorithmic}
\end{algorithm}

\isollison*

\begin{proof}
Fix an iteration $i$ of \iscollision{} (Line~\ref{alg: collision line 1}).  
The probability that no collision occurs in the $i$-th iteration is
\begin{align*}
    \Pr[\collision_i = 0]
    = \prod_{j=0}^{\preresamplesize - 1}
      \left(1 - \frac{j}{\edgecount(U)}\right).
\end{align*}

\noindent If $\edgecount(\graph) \leq \threshold$,
We upper bound the probability of no collision as follows:
\begin{align*}
    \Pr[\collision_i = 0]
    &\leq \exp\left(
        - \sum_{j=0}^{\preresamplesize - 1}
        \frac{j}{\edgecount(\graph)}
    \right)
    && \text{($1 + x \leq e^{x}$ for all $x \in \R$)} \\
    &= \exp\left(
        - \frac{\preresamplesize(\preresamplesize - 1)}{2 \edgecount(\graph)}
    \right) \\
    &\leq \exp\left(
        - \frac{\preresamplesize(\preresamplesize - 1)}{2 \threshold}
    \right)
    && \text{(since $\edgecount(\graph) \leq \threshold$)} \\
    &\leq \frac{1}{e}
    && \text{(using $\preresamplesize = \sqrt{2\threshold}$)}.
\end{align*}

\noindent If $\edgecount(\graph) \geq 8\threshold$,  we lower bound the probability of no collision:
\begin{align*}
    \Pr[\collision_i = 0]
    &\geq \exp\left(
        - \sum_{j=0}^{\preresamplesize - 1}
        \left(
            \frac{j}{\edgecount(\graph)} +
            \frac{j^2}{\edgecount(\graph)^2}
        \right)
    \right)
    && \text{($1 - x \geq e^{-x - x^2}$ for $x \in [0, 2/3]$)} \\
    &\geq \exp\left(
        - \frac{\preresamplesize(\preresamplesize - 1)}{16 \threshold}
        + \Theta\!\left(\frac{1}{\sqrt{\threshold}}\right)
    \right) \\
    &\geq \frac{1}{2}
    && \text{(using $\preresamplesize = \sqrt{2\threshold}$)}.
\end{align*}

\noindent If $ \edgecount( \graph ) \leq \threshold $, $ \Pr[ \collision_i = 0 ] \leq \frac{ 1 }{ e } $. \iscollision{} outputs $0$ when more than half of $ \collision_i = 0 $. Over $ \log ( \confidence^{-1} ) $ parallel runs, using the upper tail estimate of the Chernoff Bound (see Lemma~\ref{lemma: Multiplicative Chernoff Bound}), the probability that \iscollision{} outputs $0$ is at most $ \exp ( - \log ( \confidence^{-1} ) ) = \confidence $. Similarly, if $ \edgecount( \graph ) \geq 8\threshold $, the probability that \iscollision{} outputs $1$ is at most $ \confidence $.

\end{proof}

\begin{algorithm}
    \caption{\birthdayestimate{}($ \graph, \threshold, \approxerror, \confidence $)}\label{alg: birthdayestimate}
    \begin{algorithmic}[1]
        \Require \ciscq{} access to the $\qsetno$-partite $\hyperlen$-hypergraph $G$, error parameter $\approxerror \in (0, 1)$, confidence $\confidence \in (0, 1)$ and a threshold $ \threshold \geq 1 $.
        \Ensure An estimate of $ \edgecount( G ) $
        \State Uniformly and independently sample $ \Theta\fbrac{ \approxerror^{-1} \sqrt{\threshold} \log \fbrac{ \confidence^{-1} }} $ random hyperedges from $ G $ using \ciscq{} query.
        \State Let $ \resamplesize = \Theta\fbrac{ \approxerror^{-1} \sqrt{\threshold} \log \fbrac{ \confidence^{-1} }}  $ be the number of samples drawn and $ \collisioncount $ is the number of collisions.
        \State \Return $\binom{\resamplesize}{2}/\collisioncount $.
    \end{algorithmic}
\end{algorithm}

\birth*

\begin{proof}
    Let $\sbrac{e_1, e_2, \ldots, e_{\resamplesize}}$ be the sampled edges. or $1 \le i < j \le \resamplesize$, define the indicator random variable $X_{ij}$ to be $1$ if $e_i = e_j$ and $0$ otherwise . Then, 
    $$\Pr[X_{ij}=1] = \frac{1}{ \edgecount( \graph ) } \text{ and } \E[X_{ij}] = \frac{1}{ \edgecount( \graph ) }$$
    Define the total number of collisions by $ X = \sum_{1 \le i < j \le \resamplesize} X_{ij} $. By the linearity of expectation
    \begin{align*}
    \E[X] = \binom{\resamplesize}{2} \frac{1}{ \edgecount( \graph ) }
    \end{align*}
    Since the random variables $ X_{ij} $ are not independent. We will bound the variance using
    \begin{align}\label{eq: variance equation}
        \Var(X) &= \sum_{i < j } \Var(X_{ij}) + 2\sum_{i < j, k < \ell} \mathrm{Cov}(X_{ij},X_{k\ell})
    \end{align}
    Let us bound the first term $ \Var(X_{ij}) $ in~\eqref{eq: variance equation}
    \begin{align*}
        \sum_{i < j } \Var(X_{ij}) \leq \sum_{ i < j } \E[X_{ij}^2] \leq \binom{\resamplesize}{2} \frac{1}{ \edgecount( \graph ) } =  \E[X]
    \end{align*}
    Next, consider the covariance terms in~\eqref{eq: variance equation}. If $X_{ij}$ and $X_{k\ell}$ correspond to disjoint pairs, then they are independent and the covariance is zero. If they overlap in exactly one index (share exactly one index), then
    \begin{align*}
        \mathrm{Cov}(X_{ij},X_{j\ell}) &= \E[X_{ij}X_{j\ell}] - \E[X_{ij}]\E[X_{j\ell}]\\  
        &= \Pr[X_{ij} = X_{j\ell} = 1] - \frac{1}{ \edgecount( \graph )^2}\\
        &= \Pr[e_i = e_j = e_\ell] - \frac{1}{ \edgecount( \graph )^2 } \\
        &= 0
    \end{align*}
So, the bound on the variance is $ \Var(X)  = \sum_{i < j } \Var(X_{ij}) \leq \E[X] $. Using Chebyshev’s Inequality (Lemma~\ref{lemma: chebyshev}), we have
\begin{align}
    \nonumber\Pr[\abs{X - \E[X]} > \frac{\approxerror}{3} \E[X]] &\leq \frac{9\Var(X)}{\approxerror^2 (\E[X])^2}\\
    \nonumber&\leq \frac{9}{\approxerror^2\E[X]} &\text{$\fbrac{\Var(X) \leq \E[X]}$}\\
    \nonumber&= \frac{ 18\edgecount( \graph ) }{ \approxerror^2\resamplesize(\resamplesize - 1) } &\text{$\fbrac{ \E[X] = \binom{\resamplesize}{2} \frac{ 1 }{ \edgecount( \graph ) } }$}\\
    &\leq \confidence &\text{($ \edgecount(\graph) \leq \threshold $ and $\resamplesize = \Theta ( \sqrt{ \threshold } \log \fbrac{ \confidence^{-1} } /{ \approxerror } )$)}\label{ineq: collision bound}
\end{align}
To complete the proof, we need to show with probability at least $ 1 - \confidence $
\begin{align*}
    \fbrac{1 - \approxerror} \edgecount( \graph ) \leq \binom{\resamplesize}{2}  \frac{1}{X} \leq \edgecount( \graph ) \fbrac{1 + \approxerror}
\end{align*}
From~\eqref{ineq: collision bound}, with probability at least $ 1 - \confidence $,
\begin{align*}
    \fbrac{1 - \frac{\approxerror}{3}}\E\tbrac{X} \leq X \leq \fbrac{1 + \frac{\approxerror}{3}} \E\tbrac{X} 
\end{align*}
This implies
\begin{align*}
    \Rightarrow \frac{ \edgecount( \graph ) }{\fbrac{1 + \frac{\approxerror}{3}}}\leq \binom{\resamplesize}{2}  \frac{1}{X} \leq \frac{\edgecount( \graph ) }{\fbrac{1 - \frac{\approxerror}{3}}}
\end{align*}
For any $ \approxerror \in (0, 1) $
\begin{equation*}
    1 - \approxerror \leq \fbrac{1 + \frac{\approxerror}{3}}^{-1} \text{ and } \fbrac{1 - \frac{\approxerror}{3}}^{-1} \leq 1 + \approxerror
\end{equation*}
Since $X$ is the number of collisions. We conclude that
    \begin{align*}
        \fbrac{1 - \approxerror} \edgecount( \graph ) \leq \binom{\resamplesize}{2} \frac{1}{\collisioncount} \leq \edgecount( \graph ) \fbrac{1 + \approxerror}
    \end{align*}
which completes the proof.
\end{proof}

\section{Concentration Bounds}
We use the following concentration bounds~\cite{Mitzenmacher_Upfal_2005}.

\begin{lemma}[Chebyshev’s Inequality]\label{lemma: chebyshev}
    Let $ X $ be a random variable and $ a > 0 $, Then
    \[
    \Pr[ \abs{X - \E[X]} > a] \leq \frac{ \Var( X ) }{ a^2 }
    \]
\end{lemma}

\begin{lemma}[Hoeffding's bound]\label{lemma: hoeffding}
    Let $X_1, X_2, \ldots, X_t$ be independent random variables such that for all $i \in [t]$, $\E[X] = \mu $ and $\Pr[a \leq X_i \leq b] = 1$. Then
    \begin{equation*}
        \Pr \tbrac{ \abs{ \frac{ 1 }{ t } \sum_{i = 1}^t X_i - \mu }  \geq \approxerror } \leq 2 \exp{ \fbrac{ - \frac{ 2t\approxerror^2 }{ (b - a)^2 } } }.
    \end{equation*}
\end{lemma}

\begin{lemma}[Multiplicative Chernoff Bound]\label{lemma: Multiplicative Chernoff Bound}
    Let $X_1,X_2,...,X_t$ be i.i.d. random variables
    where $\Pr[X_i = 1] = p$ and $\Pr[X_i = 0] = 1-p$ for all $i \in [t]$, and $X = \sum_{i \in [t]} X_i$. Then, we have
    \begin{align*}
    \Pr \left[ X \leq (1-\approxerror) \E\tbrac{X} \right] &\leq \exp{\fbrac{-\frac{\approxerror^2\E\tbrac{X}}{3}}} & 0 \leq \approxerror <1\\
    \Pr \left[ X \geq (1+\approxerror) \E\tbrac{X} \right] &\leq \exp{\fbrac{-\frac{\approxerror^2\E\tbrac{X}}{2}}} & 
   0\leq \approxerror \leq 1
    \end{align*}
\end{lemma}
\end{document}